\documentclass[final]{IEEEtran}

\IEEEoverridecommandlockouts
\usepackage{cite}
\usepackage{amsmath,amssymb,amsfonts}
\usepackage{algorithmic}
\usepackage{graphicx}
\usepackage{textcomp}
\usepackage{xcolor}
\def\BibTeX{{\rm B\kern-.05em{\sc i\kern-.025em b}\kern-.08em
    T\kern-.1667em\lower.7ex\hbox{E}\kern-.125emX}}

\usepackage{url}
\usepackage{amsthm}
\usepackage{tikz}
\usepackage[
	n,
	operators,
	advantage,
	sets,
	adversary,
	landau,
	probability,
	notions,	
	logic,
	ff,
	mm,
	primitives,
	events,
	complexity,
	asymptotics,
	keys]{cryptocode}
\usepackage{cite}
\usepackage{amsmath,amssymb,amsfonts}
 \usepackage{algorithm}
\usepackage{graphicx}
\usepackage{textcomp}
\usepackage{xcolor}
\usepackage{multirow}

\usepackage{flushend}
\usepackage{hyperref}

\usepackage{wasysym}
\usepackage[flushleft]{threeparttable}
\usepackage{tablefootnote}
\usepackage{footnote}

\usepackage{subcaption}
 \usepackage{caption}
\usepackage{titlesec}
\titlespacing{\section}{0pt}{6.5pt}{6.5pt}
\titlespacing{\subsection}{0pt}{5pt}{5pt}
\begin{document}
\newtheorem{theorem}{Theorem}[section]
\newtheorem{lemma}[theorem]{\bf Lemma}
\newtheorem{proposition}[theorem]{Proposition}
\newtheorem{corollary}[theorem]{Corollary}

\newenvironment{definition}[1][Definition]{\begin{trivlist}
\item[\hskip \labelsep {\bfseries #1}]}{\end{trivlist}}
\newenvironment{example}[1][Example]{\begin{trivlist}
\item[\hskip \labelsep {\bfseries #1}]}{\end{trivlist}}
\newenvironment{remark}[1][Remark]{\begin{trivlist}
\item[\hskip \labelsep {\bfseries #1}]}{\end{trivlist}}

\newcommand{\squishlist}{
  \begin{list}{$\bullet$}
    { 
    \setlength{\itemsep}{0pt}      \setlength{\parsep}{0pt}
      \setlength{\topsep}{3pt}       \setlength{\partopsep}{0pt}
      \setlength{\listparindent}{-2pt}
      \setlength{\itemindent}{-5pt}
      \setlength{\leftmargin}{0.5em} \setlength{\labelwidth}{0em}
      \setlength{\labelsep}{0.5em} } }

\newcommand{\squishlistend}{
    \end{list}  }
      
 \newcounter{Lcount}
\newcommand{\squishlisttwo}{
\begin{list}{\arabic{Lcount}) }
{ \usecounter{Lcount}
\setlength{\itemsep}{0pt}
\setlength{\parsep}{0pt}
\setlength{\topsep}{0pt}
\setlength{\partopsep}{0pt}
\setlength{\leftmargin}{0.6 em}
\setlength{\labelwidth}{1.5em}
\setlength{\labelsep}{0.5em} } }

\newcommand{\squishlisttwoend}{
\end{list} }

 \newcommand{\squishlistthree}{
  \begin{list}{$-$}
    { 
    \setlength{\itemsep}{0pt}      \setlength{\parsep}{0pt}
      \setlength{\topsep}{3pt}       \setlength{\partopsep}{0pt}
      \setlength{\listparindent}{-2pt}
      \setlength{\itemindent}{-3pt}
      \setlength{\leftmargin}{0.5em} \setlength{\labelwidth}{0em}
      \setlength{\labelsep}{0.5em} } }

\newcommand{\squishlistthreeend}{
    \end{list} }

 \newcommand{\squishlistthreeone}{
  \begin{list}{$-$}
    { 
    \setlength{\itemsep}{0pt}      \setlength{\parsep}{0pt}
      \setlength{\topsep}{3pt}       \setlength{\partopsep}{0pt}
      \setlength{\listparindent}{-4pt}
      \setlength{\itemindent}{-7pt}
      \setlength{\leftmargin}{0.0em} \setlength{\labelwidth}{0em}
      \setlength{\labelsep}{0.5em} } }

\newcommand{\squishlistthreeoneend}{
    \end{list} }
    
\newcounter{Lcountfour}
\newcommand{\squishlistfour}{
\begin{list}{\alph{Lcountfour}) }
{ \usecounter{Lcountfour}
\setlength{\itemsep}{0pt}
\setlength{\parsep}{0pt}
\setlength{\topsep}{0pt}
\setlength{\partopsep}{0pt}
\setlength{\leftmargin}{0.8 em}
\setlength{\labelwidth}{1.5em}
\setlength{\labelsep}{0.5em} } }

\newcommand{\squishlistfourend}{
\end{list} }

\newcounter{Lcountfive}
\newcommand{\squishlistfive}{
\begin{list}{\arabic{Lcountfive}) }
{ \usecounter{Lcountfive}
\setlength{\itemsep}{0pt}
\setlength{\parsep}{0pt}
\setlength{\topsep}{0pt}
\setlength{\partopsep}{0pt}
\setlength{\leftmargin}{1.2 em}
\setlength{\labelwidth}{1.5em}
\setlength{\labelsep}{0.5em} } }

\newcommand{\squishlistfiveend}{
\end{list} }

\newcounter{Lcountsix}
\newcommand{\squishlistsix}{
\begin{list}{\roman{Lcountsix}. }
{ \usecounter{Lcountsix}
\setlength{\itemsep}{0pt}
\setlength{\parsep}{0pt}
\setlength{\topsep}{0pt}
\setlength{\partopsep}{0pt}
\setlength{\leftmargin}{1.2 em}
\setlength{\labelwidth}{1.5em}
\setlength{\labelsep}{0.5em} } }

\newcommand{\squishlistsixend}{
\end{list} }

\newcommand{\squishlistoneone}{
   \begin{list}{$\bullet$}
    { \setlength{\itemsep}{0pt}      \setlength{\parsep}{0pt}
      \setlength{\topsep}{3pt}       \setlength{\partopsep}{0pt}
      \setlength{\listparindent}{-2pt}
      \setlength{\itemindent}{-5pt}
      \setlength{\leftmargin}{0.7em} \setlength{\labelwidth}{0em}
      \setlength{\labelsep}{0.3em} } }

\newcommand{\squishlistoneoneend}{
    \end{list}  }
    
\newcounter{Lcounttwotwo}
\newcommand{\squishlisttwotwo}{
\begin{list}{\arabic{Lcounttwotwo}) }
{ \usecounter{Lcounttwotwo}
\setlength{\itemsep}{0pt}
\setlength{\parsep}{0pt}
\setlength{\topsep}{0pt}
\setlength{\partopsep}{0pt}
\setlength{\leftmargin}{1.8 em}
\setlength{\labelwidth}{1.5em}
\setlength{\labelsep}{0.3em} } }

\newcommand{\squishlisttwotwoend}{
\end{list} }

\newcommand\blfootnote[1]{%
  \begingroup
  \renewcommand\thefootnote{}\footnote{#1}%
  \addtocounter{footnote}{-1}%
  \endgroup
}

\date{}

\title{Argus: A Fully Transparent Incentive System for Anti-Piracy Campaigns (Extended Version)}

\author{
    \IEEEauthorblockN{Xian Zhang\IEEEauthorrefmark{1}, Xiaobing Guo\IEEEauthorrefmark{2}, Zixuan Zeng\IEEEauthorrefmark{3}, Wenyan Liu\IEEEauthorrefmark{4}, Zhongxin Guo\IEEEauthorrefmark{1}, Yang Chen\IEEEauthorrefmark{1}, Shuo Chen\IEEEauthorrefmark{1},\\ Qiufeng Yin\IEEEauthorrefmark{1}, Mao Yang\IEEEauthorrefmark{1}, Lidong Zhou\IEEEauthorrefmark{1}}
    \\
    \IEEEauthorblockA{\IEEEauthorrefmark{1}Microsoft Research Asia \hspace{10pt} \IEEEauthorrefmark{2}Alibaba Group \hspace{10pt}\IEEEauthorrefmark{3}Carnegie Mellon University	 \hspace{10pt}\IEEEauthorrefmark{4}East China Normal University
    \\\{zhxian, zhogu, yachen, shuochen, qfyin, maoyang, lidongz\}@microsoft.com
    \\xiaobing.gxb@alibaba-inc.com \hspace{25pt}zixuanze@andrew.cmu.edu \hspace{25pt}wyliu@stu.ecnu.edu.cn}
}

\maketitle

\begin{abstract}
Anti-piracy is fundamentally a procedure that relies on collecting data from the open anonymous population, so how to incentivize credible reporting is a question at the center of the problem. Industrial alliances and companies are running anti-piracy incentive campaigns, but their effectiveness is publicly questioned due to the lack of transparency. We believe that full transparency of a campaign is necessary to truly incentivize people. It means that every role, e.g., content owner, licensee of the content, or every person in the open population, can understand the mechanism and be assured about its execution without trusting any single role. 

We see this as a distributed system problem. In this paper, we present Argus, a fully transparent incentive system for anti-piracy campaigns. The groundwork of Argus is to formulate the objectives for fully transparent incentive mechanisms, which securely and comprehensively consolidate the different interests of all roles. These objectives form the core of the Argus design, highlighted by our innovations about a Sybil-proof incentive function, a commit-and-reveal scheme, and an oblivious transfer scheme. In the implementation, we overcome a set of unavoidable obstacles to ensure security despite full transparency. Moreover, we effectively optimize several cryptographic operations so that the cost for a piracy reporting is reduced to an equivalent cost of sending about 14 ETH-transfer transactions to run on the public Ethereum network, which would otherwise correspond to thousands of transactions. With the security and practicality of Argus, we hope real-world anti-piracy campaigns will be truly effective by shifting to a fully transparent incentive mechanism.  
\end{abstract}

\blfootnote{\IEEEauthorrefmark{2}\IEEEauthorrefmark{3}\IEEEauthorrefmark{4}Work done during employment (Xiaobing Guo) and internship (Zixuan Zeng and Wenyan Liu) at Microsoft.}
\setlength{\baselineskip}{12pt}

\section{Introduction}\label{sec:intro}

Intellectual property is one of the most valuable assets for present-day companies, especially in the software, movie, gaming and digital publishing industries. Anti-piracy is a long-lasting and heavily invested effort, because piracy impacts the fundamental business models of these industries. Anti-piracy has the legal aspect and the technological aspect. The former crucially depends on the latter to collect credible and undeniable evidences, so that appropriate legal actions can be taken against the infringers. For example, if evidences prove that the number or the retail value of the pirated copies exceed a certain threshold during any 180-day period, according to the Title 17 of the United States Code, infringers shall be imprisoned for a maximum of 5 years, or fined a maximum \$250,000, or both \cite{title17}.

Since piracy is fundamentally about disseminating copyrighted contents outside legitimate distribution channels, a central question about anti-piracy is how to incentivize people in the open population to report pirated copies. Industrial alliances (BSA \cite{bsa}, FACT\cite{fact}, SIIA\cite{siia}) and companies (e.g., Custos \cite{custos}, Veredictum \cite{veredictum}) have offered big amounts of bounties for piracy reporting. For example, the Business Software Alliance (i.e., BSA \cite{bsa}), whose members include Apple, IBM, Microsoft, Symantec and many others, posted a \$1-million bounty for reporting. However, the approach is not yet effective and is questioned/criticized by the public, mainly due to the lack of transparency \cite{bsanews}. For example, it is unclear whether the \$1-million total bounty is simply a marketing gimmick, as BSA had only rewarded a small fraction of it in a long period of time. Also, BSA's neutrality is really problematic, as its members are copyright holders, who may not represent the best interest of the informers in the public. Moreover, it is unclear how BSA evaluates the credibility of the piracy reports or whether they are strong enough against an infringer's repudiation. Obviously, opaqueness about incentive, fairness, and credibility-criteria seriously limit the effectiveness of these anti-piracy campaigns.

We envision that a methodological progress for anti-piracy campaigns can be made once they are formulated as a decentralized computing problem. It is promising to advance the status quo by distributed-system technologies, especially those about incentive model (NF-Crowd\cite{nf}, Arbitrum \cite{arbitrum}, Hydra \cite{hydra}), consensus mechanism (PBFT-Hyperledger \cite{hyperledgerpbft}, Algorand \cite{algorand}), secure messaging (Decentralized  release \cite{release}, Hyperpubsub \cite{hyperpubsub}, Zerocash \cite{zerocash,zkpcommit}) and Sybil resistance (Arbitrum \cite{arbitrum}, TrueBit \cite{truebit}). Toward the vision, we have built a concrete system named Argus\footnote{Argus (an abbreviation for Argus Panoptes) is a many-eyed giant in Greek mythology, which comprehensively traces misbehaviors.}. The design is based on a clear problem statement and a set of properties as objectives, which are explained next.

\vspace{5pt}
\noindent {\bf Problem statement}. We describe the anti-piracy problem using the following terminology. An {\em owner} is the one who owns the copyrighted content (e.g., a film maker). The content is distributed through a controlled channel to a set of {\em licensees} (e.g., cinemas and film critics). Some licensees may leak their copies of the content, which leads to many pirated copies in the open population. These licensees are called the {\em infringers}. The anti-piracy system's goal is to incentivize people in the open population to report the pirated copies to the system. We refer to these people as the {\em informers}. 

The challenge in the anti-piracy problem is that the interests of these roles are different or even conflicting. Specifically, the owner's goal is to identify infringers and assess the severity of the infringement. The owner wishes that, by giving a financial incentive (e.g., a bounty) to the open population, as many good-faith reports as possible can be received. However, the motivation of informers is not always aligned with the owner. It is only reasonable to assume that informers are financially motivated \cite{arbitrum,nf}, like black-hat security researchers motivated by bug bounties \cite{hydra}. Not surprisingly, an infringer's best interest is to refute the credibility of an evidence by arguing that it could be fabricated by an informer or the owner. 

Because of the conflict of interests, an unbiased solution would require a contract that was agreed upon by all these roles. But who should be the executor of the contract? One possibility is to introduce into the problem an “executor” role, as in the BSA situation described earlier, but the role is really undesirable because its neutrality is hard to be assured in reality (e.g., even big companies like Facebook and Google had controversial practices that put their neutrality in doubt \cite{neutral}). Our work is to explore the feasibility of an open contract of which the execution is fully transparent to the public, without an additional role as the trust basis. 

\vspace{5pt}
\noindent {\bf The Argus system}. In this paper, we present the design, implementation and evaluation of the Argus system. To the best of our knowledge, it is the first public anti-piracy system which (1) does not hinge on any ``trusted" role; (2) treats every participant fairly (in particular, it is resilient to greed and abuse, and resolves conclusively every foreseeable conflict); and (3) is efficient and economically practical to run on a public blockchain (e.g. it achieves an impressive off-chain throughput of 82.6 data-trades per second per machine, and incurs only a negligible on-chain cost  equivalent to sending 14 ETH-transfer transactions per report on the public Ethereum blockchain). 

The four pillars in the Argus design are \textit{full transparency}, \textit{incentive}, \textit{information hiding} and \textit{optimization}. These are the main focuses to be elaborated on in this paper. It is worth noting that they are not four problems to be solved individually, but integral aspects in one coherent design. We highlight some properties of Argus, which represent some of our core innovations:

\textit{Incentivizing good-faith informers}. A fundamental challenge is about the interest of informers, who are anonymous people in the open population. The owner's interest is to collect good-faith reports so that the severity of the infringement can be accurately estimated. However, each individual informer's interest is to maximize his own reward. What prevents an informer from creating multiple identities to make multiple reports, so that he gets multiple rewards but causes the owner's estimation to be inflated? Note that an attack using multiple forged identities is often referred to as the Sybil attack \cite{arbitrum,truebit}. In Argus, the incentive model ensures that the total reward of the informer and all his Sybils is less than the reward he would get without forging the Sybils. In other words, our model disincentivizes Sybil attacks, so the informers' interest is aligned with the owner's. In addition, our model is superior to previous models because of several other properties for better incentives (Section \ref{sec:quantify}).

\textit{Information hiding for report submission}. Because Argus runs on a public ledger, its execution is fully transparent to everybody \cite{zerocash,zkpcommit}. It is crucial that an informer is unable to resubmit any report previously submitted by somebody else. For this reason, Argus' report submission protocol is based on Multi-period Commitment Scheme, which gives a ``zero-knowledge'' style guarantee, i.e., a submission only proves that the informer has a copy of the content without disclosing other information. Compared to traditional commitment schemes, our scheme leaks no useful information even in the reveal phase while avoiding the heavy cost of zero-knowledge proof (Section \ref{sec:zkw}). 


\textit{Strong accusation against infringer}. An owner's accusation against a licensee is always subject to an inherent paradox – since both have the leaked copy, why can't the licensee refute the accusation by arguing that the infringer may be the owner himself? We believe that the only solution to circumvent this paradox is to resort to a probabilistic argument. Argus uses Oblivious Transfer~(OT) to ensure that the false accusation is bounded by a probability $\phi$, which can be arbitrarily small. Hence, the accusation is very hard to refute. Moreover, we improve the efficiency of leveraging OT, which carefully considers the scalability limitation of distributed ledgers \cite{hyperledgerpbft,algorand} (Section \ref{sec:detect}).

\vspace{5pt}
\noindent {\bf Contributions.} Our contributions are as follows:
\squishlistoneone
\item We formulate anti-piracy as a problem about consolidating different interests of multiple roles, including informers in the open population. We also clearly state the design objectives of anti-piracy solutions, which give a foundation for this work and future research.
\item The Argus contract is fully transparent --- no role is considered as the trust base. This is a significant advancement. In addition, our approach is systematic. Because of the clarity on the design objectives, we are able to deduce the general form of the incentive model and identify all the unavoidable technical challenges. Because these challenges are general, solving them in Argus will have a far-reaching impact in the broad problem space. 
\item To achieve full transparency and a number of novel objectives, Argus needs to implement sophisticated cryptographic operations as contract code, rather than native code. Optimizations are a vital effort in the design of Argus. We show that, if existing cryptographic operations were adopted without optimization, the cost would equal sending thousands of transactions (as opposed to 11 transactions in Argus), which would make the solution economically unreasonable. 
\squishlistoneoneend

\noindent {\bf Roadmap}. The rest of the paper is organized as follows. Section \ref{sec:basicargus} gives an overview of Argus to address how we leverage a transparent contract to achieve the mutual trust and fairness between owner, licensees and informers, along with related primitives. In Section \ref{sec:quantify}, Section \ref{sec:zkw} and Section \ref{sec:detect}, we further figure out that those primitives should be optimized to achieve four pillars with better practicality. Then, we brief our implementation with optimizations in Section \ref{sec:construct}. In Section \ref{sec:eval}, the security analysis and performance evaluation of Argus are provided, followed by related work in Section \ref{sec:relate} and a conclusion. The Appendix provides details of mathematical deduction, protocol implementations and security analysis.

\section{Overview of Argus} \label{sec:basicargus}
In this section, we give an overview of the anti-piracy solution and Argus contract. We assume familiarity with contract \cite{solidity} and cryptographic primitives such as oblivious transfer \cite{np}, zero-knowledge proof \cite{zksnark} and commitment scheme \cite{commitment}. 

\subsection{The Argus Contract} 
Figure \ref{fig:overview} illustrates the important elements in the Argus contract and how different roles interact with it. 
An instance of the Argus contract is created by the owner of a copyrighted content. We denote $M$ as the number of licensees and list the data fields used in the Argus system in Table \ref{tab:data}.
\begin{table}[b]
\vspace{-15pt}
\footnotesize
\centering
\caption{Data fields used in the Argus system}
\setlength{\tabcolsep}{1pt}
 \renewcommand{\arraystretch}{1.1}
\begin{tabular}{|l|l|}
\hline
Data fields           & Description                                                                                                                                                              \\ \hline
\textrm{LicenseeStatus[$M$]} & \begin{tabular}[c]{@{}l@{}}To store the status of every licensee. Each \\ status belongs to set \{NORMAL, ACCUSED, \\ GUILTY, EXONERATED\}\end{tabular} \\ \hline
\textrm{OTEList[$M$]}        & \begin{tabular}[c]{@{}l@{}}To store the OTEvidence of every OT between\\ owner and Licensee\end{tabular}                                                            \\ \hline
\textrm{ReportNumber[$M$] }  & \begin{tabular}[c]{@{}l@{}}To indicate the number of leaked copies from\\ every licensee by counting informers' reports\end{tabular}                                      \\ \hline
\end{tabular} \label{tab:data}
\end{table}

\begin{figure}[t]
    \centering
    \includegraphics[width=0.43\textwidth]{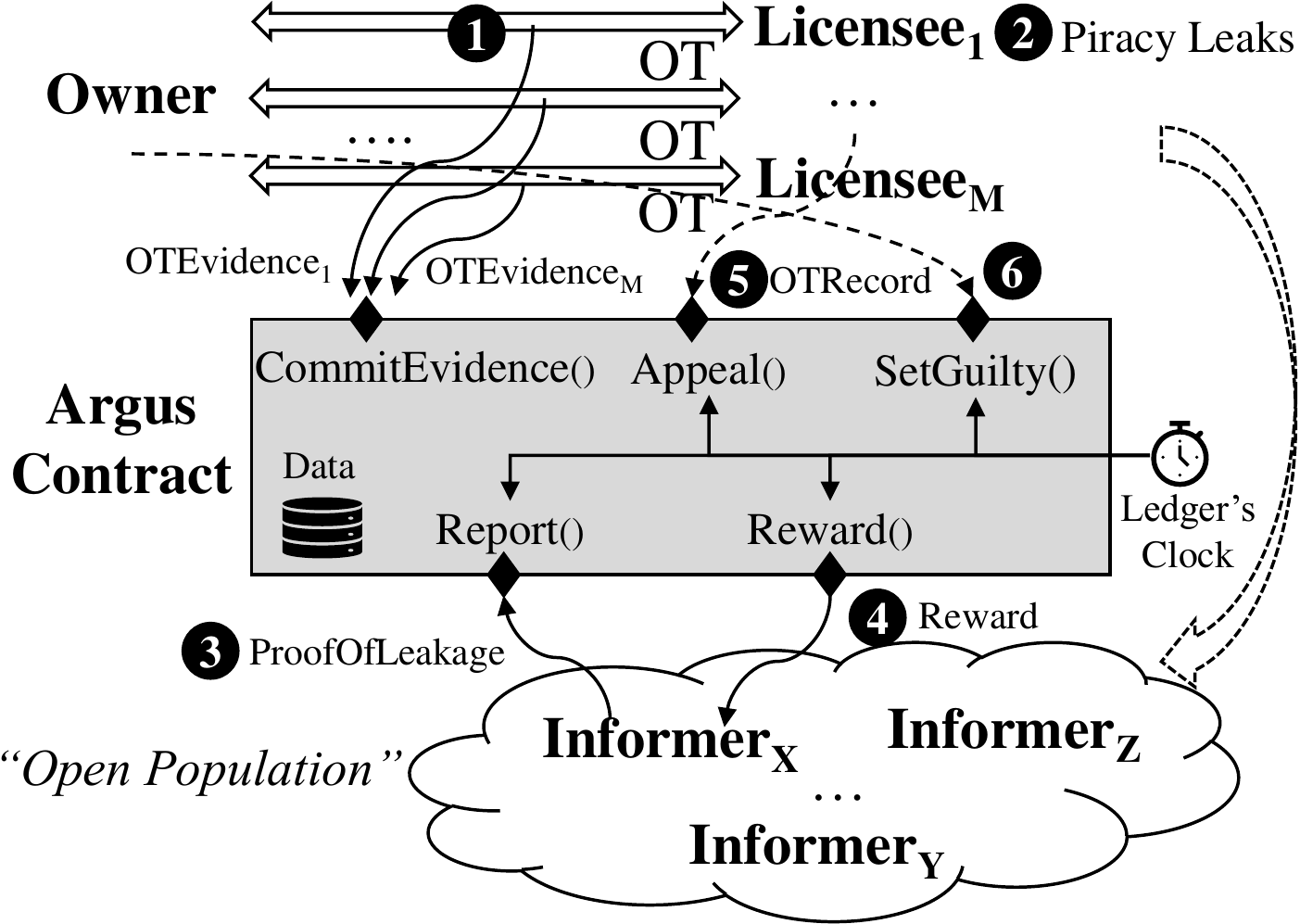}
    \vspace{-5pt}
    \caption{Overview of the Argus system. }
    \label{fig:overview}
    \vspace{-20pt}
\end{figure}

When the owner distributes the content to the licensees, she generates a large number of watermarked copies, e.g., $10000*M$ copies. In other words, each copy is embedded with a unique secret string. In this work, we assume that nobody can remove the watermark from a copy without badly deteriorating the content quality. Despite this assumption, we will describe in Section \ref{sec:model} our insights about improving the robustness of watermark.

A licensee retrieves a copy from the owner through oblivious transfer (OT), which ensures: (1) the owner does not know which of the watermarked copies is retrieved; (2) the licensee knows nothing about other copies except the retrieved one. The OT procedure is performed by private computations on the owner side and the licensee side. One of our important features added to the existing OT protocol is to produce a data record called {\em OTRecord} on the licensee side, which is needed in case the licensee appeals against an accusation. In addition, the OT scheme produces another piece of data called {\em OTEvidence}, which opaquely represents the existence of the current OT procedure. Every OTEvidence is submitted to the contract via function \texttt{CommitEvidence(OTEvidence)}. The set \texttt{OTEList[M]} defined in the contract stores all OTEvidences. 

Suppose a copy of the content is leaked out into the open population. In Figure \ref{fig:overview}, we assume that informer$_X$ gets the copy and wants to report. The informer should first extract the secret string from the watermarked copy. The reporting is to fulfill an information-hiding procedure to show informer's acquaintance of the secret string, which is via function \texttt{Report(ProofOfLeakage)} of the contract and increases an element of \texttt{ReportNumber[M]}.  
Function \texttt{Report()} can be either implemented via commitment scheme or zero knowledge proof.
Because \texttt{Report()} is an information-hiding procedure, other people in the population cannot learn anything about the watermarked copy reported by informer$_X$, thus cannot report about the same copy unless they actually have it.
Nevertheless, the information-hiding submission does not prevent informer$_X$ from creating multiple Sybil identities to submit multiple reports about one watermarked copy. Function \texttt{Reward()} of the Argus contract implements an incentive model that causes the total reward obtained through the Sybil attack to be less than what the informer would normally receive. We will explain the incentive model in Section \ref{sec:quantify}. Function \texttt{Reward()} is periodically invoked according to the ledger's clock so that informers can get timely rewards.

Once a report is received by the contract, the reported copy will reveal which licensee is accused and update\texttt{ LicenseeStatus[M]}. The status of the licensee is changed from NORMAL to ACCUSED. The licensee has a time period for appeal. If he does not call function \texttt{Appeal(OTRecord)} within the period, function \texttt{SetGuilty()}, also invoked by the owner according to the ledger's clock, changes the licensee's status to GUILTY. If he calls \texttt{appeal(OTRecord)} within the period, the OTRecord will reveal which watermarked copy the licensee received via OT. If it is the same copy revealed by the report, the licensee's status is changed to GUILTY. Otherwise, it is changed to EXONERATED. Once exonerated, this licensee is exonerated for this copyrighted content, because the identity of his copy has been disclosed. 

There are different ways to implement the Argus contract, but every solution should include three elements: (1) an {\bf\em incentive model}, (2) an {\bf\em information-hiding submission scheme}, and (3) an {\bf\em OT scheme}. Before presenting these elements in following sections, we introduce the trust assumptions of Argus below.

\subsection{Trust Assumptions} \label{sec:model}
We make three trust assumptions for the design and evaluation of Argus.
\squishlistoneone
\item{\bf Robust watermarking.} We assume watermarking cannot be compromised without considerable degradation of the content's value. Existing work has shown success of watermark robustness in certain domains. For example, the image watermark can protect against attacks such as splitting, sampling, filtering, image compression \cite{image}. Although novel attacks such as oracle attack \cite{sensitivity} and collusion attack \cite{bscode} emerge, watermark techniques keep evolving with new countermeasures \cite{sensitivity2,tardos}. We see this as a separate research concern.
\item{\bf Financially motivated informers.} We assume that a financial reward is the only motivation for every informer \cite{arbitrum, truebit} . This implies that, if an action results in a loss of reward, no informer will take the action. 
\item{\bf Trusted blockchain.} We assume the trustworthiness of the blockchain. There are technologies to enhance security at the ledger layer \cite{pow,pos} and the smart contract layer \cite{teether,oyente,zeus,paymentprotect}. We also see this as a topic complementary to Argus.
\squishlistoneoneend

\section{Incentive Model}\label{sec:quantify}
As mentioned in the introduction, an incentive-compatible mechanism is required in Argus. It is a challenge because of the open population --- the owner does not know the real-life identities of informers. The interests of the owner and the informers are different: the owner wants to receive good-faith reports, and the informers (both honest ones and greedy ones) want to get financial rewards. The goal of the incentive model is to consolidate these different interests.
The design of our incentive model is inspired by some models in the literature \cite{arbitrum,truebit}. However, because we have the target problem clearly formulated (as in the introduction), we are able to deduce a general form of the incentive model. Our approach is superior for three reasons: (1) it is unclear how the incentive models in the literature were obtained. They seem more of a ``creative art'' than a result of a disciplined design; (2) our general form encompasses all models that satisfy the incentive objectives of all roles. The existing models in the literature are simplified special cases of our general form; (3) our model ensures extra desirable properties, such as {\em timely payout} and {\em guaranteed amount}, which we see as very important aspects of informers' incentive.
\subsection {The Objectives of Incentive Models} \label{sec:modelobj}
Before introducing our general-form model, it is worth making explicit the objectives of a desirable incentive model.
\squishlistoneone
\item First, the model must {\em disincentivize Sybil attacks}. It is easy to understand that, due to the open population, any informer can create multiple identities (i.e., Sybils), so it is not possible for the owner to detect Sybil attacks. Hence, an objective of the incentive model is to disincentivize them, so that the total reward of all Sybils of a duplicate report is lower than the reward of a single unique report. With this property, an informer (who is assumed to be financially motivated) will not inflate the owner's counts in \texttt{ReportNumber[]}.
\item Second, it is the owner's interest to {\em incentivize timely reports}, so an informer reporting earlier should be rewarded more than another one reporting later. Essentially, under the incentive model, informers are competitors racing against time. Nothing can be gained by delaying a report.
\item Third, it is the informers' interest to {\em get timely payouts} and {\em guaranteed amounts}. In all existing models, informers need to wait until the end of the campaign to know the amounts and get the rewards. We believe that a good incentive should reward an informer shortly after the report is confirmed valid. The time to reward should be independent of the campaign's duration. 
 \squishlistoneoneend

\subsection{Deducing the Reward Function from the Objectives}
We formulate the incentive model as a reward function $B(I_i,n)$, which denotes the bounty value for the $i$-th successful informer $I_i$ (informers are chronologically ordered) when the total number of informers is $n$. Previous proposals have given a special form of $B(I_i,n)$ as $c*2^{-n+1}$ ($c$ is a constant denoting the total bounty value) under certain constraints \cite{arbitrum}, but they have not formalized the properties of $B(I_i,n)$ {\em in a general form}. 
For example, the work \cite{arbitrum} only gives the special form with the assumption that $B(I_i,n)=B(I_j,n) (i \neq j)$, i.e., every informer's reward is equal. 

In this paper, we go through the process of deducing the reward function from the objectives. The mathematical definitions of the properties corresponding to aforementioned objectives are as follows:
\squishlistoneone
\item {\bf Sybil-proofness.} Arbitrary subset of informers get a reduced total reward if the total number of submissions increases: $ \sum_{i \in S_m} B(I_i,m) \ge \sum_{i \in S_k} B(I_i,k) $,  where $m \le k$ are two positive integers and $S_m \subseteq S_k$ are two arbitrary subsets of $\set{1,\ldots,m}$ and $\set{1,\ldots,k}$, respectively. 
\item {\bf Order-awareness.} The earlier the informer reports, the more bounty he/she obtains: $B(I_i,n) \ge B(I_{i+1},n)$.
\item {\bf Timely payout.} Each informer is rewarded in an amortized style: $B(I_i, n) = B_1(i)+B_2(n)$, where $B_1(i)$ is only related to $i$, so the reward can be paid immediately.
\item {\bf Guaranteed amount.} An informer, upon a confirmed reporting, is guaranteed a minimal reward amount $c_{i}>0$, i.e., his/her total reward will not be under $c_{i}$ as the number of informers increases: $\lim_{n \rightarrow \infty} B(I_i, n) \ge c_{i}$.
\squishlistoneoneend

Note that except for the first property, the others are not achieved by reward function $B(I_i,n) = c*2^{-n+1}$ from previous work. We ascribe the limitations of previous incentive models to the lack of a general-form deduction, which is addressed in our work. Due to the page limit, we move the mathematical deduction for the general form and the process of enriching properties to Appendix \ref{sec:math}. We only highlight the final expression of $B(I_i,n)$, which can be formalized as Theorem \ref{th:expr}:

\begin{theorem}[\bf The expression of $B(I_i,n)$] \label {th:expr}
To satisfy Sybil-proofness, Order-awareness, Timely Payout and Guaranteed Amount, $B(I_i,n)$ should satisfy following equation:

\begin{equation} \label{eq:gf}
 B(I_i,n) = -\xi_i+\sum_{j=i+1}^n \xi_j+c*2^{-n+1}
 \vspace{-15pt}
\end{equation}
{\small
 \begin{eqnarray} \label{eq:p1}
 \nonumber \textrm{\normalsize where }
  \xi_{i+1}=\sum_{j=1}^{i} 2^{j-i}*\Delta_j\  ,\quad
  \sum_{j=1}^{\infty} 2^{j}*\Delta_j \le c ,\quad
  \set{\Delta_i} \in  \mathbb{R}^{*}_{\ge 0}
 \end{eqnarray}
 }
 \end{theorem}

\subsection{Plugging in real-world numbers}

In the introduction, we describe the \$1-million USD campaign opaquely run by the BSA \cite{bsanews}. After a long period of time, only a small percentage of the total amount was paid out to informers. This is obviously a low incentive to the public. We now analyze the outcome if Argus is used for the \$1-million USD campaign. We instantiate the parameters $\Delta_i=2^{-i}*c/l \textrm{ } (\textrm{for } i=1,\ldots,l) \textrm{ or } 0 \textrm{ }(\textrm{for } i>l)$ in Theorem \ref{th:expr} where $l$ can be an arbitrary positive integer. For simplicity, we set $l=20$ in this work, which is a typical boundary to classify copyright infringement \cite{title17}. With this setting, we compare Argus with previous work \cite{arbitrum,truebit}, which are also Sybil-proof models.

Figure \ref{fig:instance} shows the reward amount (in the log scale) that each of the n informers will get in our incentive model (the upper diagram) and the previous model (the lower diagram). Every line in the lower diagram is horizontal, because the previous model does not have order-awareness. As a result, all informers get the same reward, and the reward is exponentially decreased with n, in order to ensure Sybil-proofness. There is no guaranteed amount when a report is confirmed. When the campaign ends, even an early informer may find the reward almost zero if there are many later informers. It is a problematic model to incentivize people.

The upper diagram shows our model. The line of $n=\infty$ corresponds to the guaranteed amounts for the informers. Our lines are also affected by n, but not as drastically as in the previous work. As n increases, the lines become closer to the $n=\infty$ line (note that the n=100 line is visually overlapped with it), suggesting that every existing informer loses some reward when a new informer joins. The design of our reward function ensures that the loss of reward is big enough so that no informer wants to fake a Sybil identity to get another reward.  

Compared to the (problematic) incentive shown in the lower diagram of Figure \ref{fig:instance}, the reward function of Argus is superior in all objectives we set in the beginning of this section. 

\begin{figure}[thbp]
    \vspace{-10pt}
    \centering
    \includegraphics[width=0.43\textwidth]{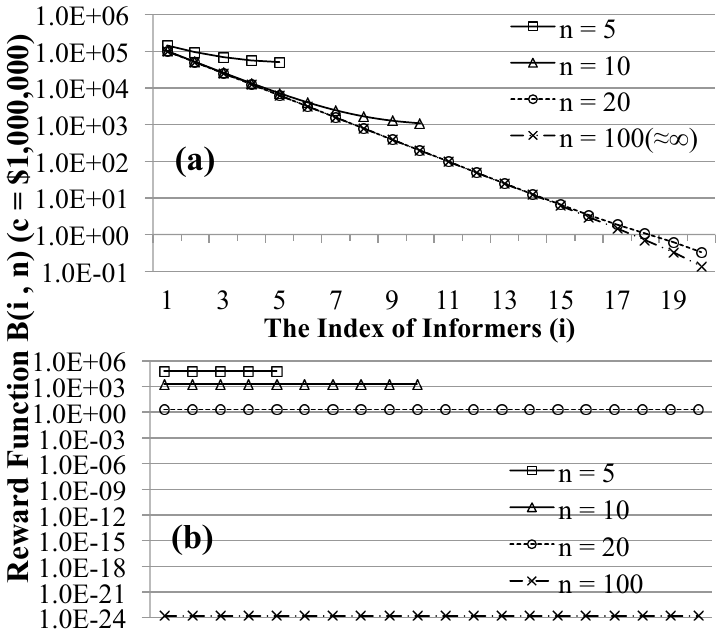}
    \caption{The reward of first twenty informers in (a) our incentive model and (b) previous incentive model \cite{truebit, arbitrum}, respectively (when total number of informers $n=5,10,20,100$). }
    \label{fig:instance}
    \vspace{-0pt}
\end{figure}


\noindent\textbf{Further improving the incentive.} It is worth to acknowledge that Sybil-proofness imposes an intrinsic property of the reward function --- it must decrease exponentially to foil a Sybil attack (see Corollary \ref{co:1} for details). The previous model, our model and future models are all constrained by this bound. For example, in our model, if n=20, the $1$st, $5$th, $10$th and $18$th informers get the reward amounts of \$ $100000$, \$$6250$, \$$196$ and \$$1.07$, respectively. For informers who report too late, if they inspect the blockchain to figure out their indices, they may not see enough incentives. Hence, to further improve the incentive model, a new objective is to conceal the actual number of submissions in blockchain transactions with techniques such as the unlinkable anonymous payment\cite{zerocash} and the submarine commitments \cite{hydra}. When this is achieved, every informer will have the equal hope that he/she may win the first (largest) reward, since the blockchain log reveals no information in this regard. We leave the design and implementation as future work.

%

\section{Information Hiding for Report Submission} \label{sec:zkw}
As mentioned in the introduction, full transparency is a major advancement of Argus. To achieve it, our design requires an effective strategy for information hiding. Specifically for the procedure of report submission, Argus needs to ensure that, although everybody in the open population can see the interactions between an informer and the Argus contract, nobody other than the informer can replay the interactions effectively. (Note that the informer himself will not replay due to the Sybil-proof property of the incentive model discussed in Section \ref{sec:quantify}.) 

\subsection {The objectives of the information-hiding submission} \label{sec:zkpobj}

Before designing the information-hiding submission procedure, we first consider the interests of different roles, and specify the objectives to be achieved:

\squishlistoneone
\item First, the owner's interest is to get an accurate count of the number of piracy copies. The count should not be inflated, meaning that no replay attack as described above should be possible.
\item Second, the owner wants to run a bounty campaign for a long duration, e.g., 3-6 months, so that the magnitude of the piracy infringement can be evaluated more thoroughly. However, an informer's interest is to get a timely payout after a report is submitted successfully.
\item Third, both the owner and the informers want the submission to be concluded in a short period of time. Also, the submission should be efficient. Since Argus runs on a public blockchain, the gas consumption of the procedure is required to be very low.
\squishlistoneoneend

\subsection {Previous proposals and their limitations}
Previous work have proposed commitment scheme \cite{hydra,commitment} and zero-knowledge proof \cite{zkbounty,arbitrum} to eliminate the replay attack. However, in Argus's scenarios, they may either conflict with the interests of owner/informer or encounter serious performance problems:
\squishlistoneone
\item In traditional commitment scheme, the informer submits a commitment of the report (e.g. Hash) in the first period (``commitment phase'') and reveals the report during a following period (``reveal phase'').  Informers' reports are not accepted in the second phase unless a corresponding commitment exists in the first phase. Therefore, the replay attack cannot succeed since the attacker cannot generate valid commitments in the first phase. However, setting the length of the commitment phase has a dilemma because of the different interests of the owner and the informers (definitions in Section \ref{sec:zkpobj}): (1) if the phase is too long, it is not timely for the owner to validate infringers and take subsequent actions, and the informer needs to wait for a long time to get the bounty; (2) if it is too short, the number of the reported pirated copies cannot accurately reflect the severity of the infringement. 
\item In a zero-knowledge proof (ZKP) scheme, the informer can generate a proof to show the possession of a pirated copy. In the ZKP scheme, the bounty can be paid as soon as the contract verifies the proof, which is what we desire. However, the ZKP scheme has a prohibitively high performance overhead and gas cost (see Section \ref{sec:eval}).  
\squishlistoneoneend

\subsection {Multi-period Commitment Scheme}
To achieve the objectives with good performance, we propose a novel technique called multi-period commitment. The scheme can be considered as an extension of the traditional commitment scheme, but has a good performance. Meanwhile, it has the advantage similar to the zero-knowledge proof --- the owner can still set a long bounty campaign period, but confirm every report almost in real time. 

Our scheme allows multiple commit-and-reveal phases so that there are sufficient time windows for the informers to submit piracy reports. For a desired length of collection period $T$, the owner can divide $T$ into $K$ {\bf\em sub-periods} $\set{T_1,\ldots,T_K}$. Each sub-period $T_i \quad (1<i<K)$ is the $i$-th commitment phase and also the $(i-1)$-th reveal phase. In other words, an informer can claim a bounty in $T_{i+1}$ ($i$-th reveal phase) by revealing if the corresponding commitment is submitted in $T_{i}$ ($i$-th commitment phase). 

However, dividing into periods introduces a problem: informers can replay the process of commit-reveal in sub-periods $T_i$ and $T_{i+1}$ to later sub-periods $T_j$ and $T_{j+1} \quad (j>i)$. To defend again this kind of replay attack, we introduce a time stamp into the formula of commitment and process of verification: if we denote the piracy report as $X$ and hash function as $\mathcal{H}$, and there is a predefined list $L[\cdot]=\set{\mathcal{H(H(}X||1)),\ldots,\mathcal{H(H(}X||K))}$ in the contract. Then, the commitment $cm$ submitted in $T_i$ can be $\mathcal{H(H(}X||i)||n)$ where ``$||$'' denotes concatenation and ``$n$'' denotes a randomized nonce. In corresponding reveal phase $T_{i+1}$, $rv=\mathcal{H}(X||i)$ and $n$ should be submitted to the contract for verification that $\mathcal{H(}rv||n)=cm$ and $\mathcal{H(}rv)=L[i]$. By this reinforcement, the aforementioned replay attack cannot pass the verification ``$\mathcal{H(}rv)=L[j]$'' in later sub-periods $T_{j+1} \quad (j>i)$.

The multi-period commitment scheme achieves the objectives given in Section \ref{sec:zkpobj}:
\squishlistoneone
    \item The multi-period commitment scheme foils the replay attack, thus meets the first objective.
    \item The multi-period commitment scheme supports an arbitrary length of collection period $T$, thus meets the second objective.
    \item With a sufficiently large $K$, each sub-period can be short enough\footnote{We will show the storage overhead of setting a large $K$ in Appendix \ref{sec:imple}.}. Thus, informers can reveal reports and get their rewards within a short interval after commitments. The owner gets quick confirmations of the infringers. This achieves the third objective.
\squishlistoneoneend
\section{Guarding against infringer's repudiation}\label{sec:detect}

No role in Argus, even the owner, is assumed trusted. This presents a challenge: when the owner accuses a licensee for leaking a copy, the licensee can refute the accusation by arguing that the copy could have been leaked out by the owner himself. To resolve the dispute, Argus must make the evidence of the accusation so convincing that {\em the probability of the accused infringer being an innocent licensee is extremely small}. Hence, a true infringer's attempt to repudiate will be unsubstantiated.

%

To approach the objective, we use a $1$-out-of-$N$ Oblivious Transfer (OT) \cite{np,otorigin} protocol to achieve this goal. The $1$-out-of-$N$ OT protocols were used for data sharing \cite{ot1,ot2,ot3,ot4}: the owner generates $N$ different copies of data (e.g. via watermarking) and plays OT protocol with the licensee. Then, the licensee can obtain only one copy without owner's knowing which one. Thus, owner can infer the chosen version with a pirated copy.  When there is a dispute between a licensee and the owner, they can submit messages occurred in OT to a {\em trusted judge} for resolving disputes. There is only a $\frac{1}{N}$ probability that the successfully accused licensee is innocent. 

A goal of Argus is not to have any trusted role, so the Argus contract has to implement the functionalities of the ``judge'' on a public blockchain. However, this may introduce a big bandwidth overhead: messages incurred in OT are proportional to $O(N)$. To achieve a desirable security level with a large $N$ (e.g. $10,000$), existing solutions introduce enormous on-chain overhead (e.g. bandwidth, execution, storage). 

To greatly reduce the overhead, we introduce {\em $O(1)$-Appeal} which only incurs $O(1)$ on-chain messages and operations. $O(1)$-Appeal has two properties:
\squishlistoneone
\item {\bf Obliviousness}. It is the property of $1$-out-of-$N$ OT \cite{np,otorigin}: (1) the licensee can arbitrarily choose and obtain $1$ data from $N$ candidate data but cannot know the unchosen data; (2) the owner does not know which data are chosen by the licensee. This property guarantees that the probability to successfully incriminate an innocent licensee is $\frac{1}{N}$ and thus an infringer is hard to deny accusation with a large $N$ (e.g. $10000$). 
\item {\bf Non-repudiation}. When the licensee is accused, the licensee can appeal by showing committed records. When there is a dispute, neither the owner nor licensees can deny which copy the licensee had chosen in the previous OT protocol. The contract is able to give a conclusive answer.
\squishlistoneoneend

\subsection{Constant-Size-OTRecord Appeal ($O(1)$-Appeal)}

The protocol of $O(1)$-Appeal is shown in Figure \ref{fig:ot}, which includes four sub protocols: $\mathsf{Initilize}$, $\mathsf{GenerateEvidence}$, $\mathsf{TransferData}$, $\mathsf{Appeal}$. The first three sub protocols are very similar\footnote{The only difference is that we add a step (step 2) in $\mathsf{GenerateEvidence}$.} to those in \cite{np}, while the fourth sub protocol is our new invention that incurs only an $O(1)$ on-chain cost. Unlike in Section \ref{sec:basicargus}, the owner does not have to commit OTEvidence (i.e. $R$) to the contract. Instead, $R$ can be signed and kept locally.
For simplicity, we assume that the owner and the licensee do not abort during the procedure (e.g. this can be ensured by using the state channel technology \cite{plasma} or the fair exchange protocol \cite{fairswap}). 

\begin{figure}[]
\begin{center}
\fbox{
\procedure[syntaxhighlight=false,mode=text,width=0.44\textwidth,codesize=\footnotesize]{OT with $O(1)$-Appeal}{
\squishlist
\vspace{5pt}
  \item {\bf Public parameters:} field $\ZZ_q$, generator $G \in \mathbb{G}$  
  \item{\bf Owner input:} $N$ versions of data $\set{D_i} (i=1,\ldots,N)$, Owner's private key $sk^O$ 
  \item {\bf Licensee input:}  Licensee's private key $sk^L$  
  \item {\bf Sub Protocols}:
  \squishlistthreeone
  \item {$\mathsf{OT.Initialize}$:} 
  \squishlisttwo
  \item Owner randomly generates and signs $N$ elements $\set{P_1, \ldots, P_N} \in \mathbb{G}^N $ and samples a random number $s \in \ZZ_q$. 
  \item Owner publishes $a_s=s\cdot G$, $\set{P_i}$ and keeps $\{P'_i\}\gets\{s \cdot P_i\} (i=1,\ldots,N)$ locally. 
  \item Licensee samples and keeps two secrets $r \in \ZZ_q$, $l \in [N]$. 
  \squishlisttwoend
  \item {$\mathsf{OT.GenerateEvidence}$:}
    \squishlisttwo
   \item Licensee signs and sends $R = P_l- r \cdot G$ to Owner. 
   \item Owner signs and sends $\sig_{sk^O}(\sig_{sk^L}(R))$ to Licensee.
   \squishlisttwoend
  \item {$\mathsf{OT.TransferData}$:} 
  \squishlisttwo
   \item Owner computes $R' \gets s\cdot R$, $Q_i \gets P'_i-R'$ and sends $\{E_i\} \gets \{\mathcal{H}(Q_i, a_s, i) \oplus D_i\} (i=1,\ldots,N)$ to Licensee. 
   \item Licensee gets $D_l= E_l \oplus \mathcal{H}(r \cdot a_s, a_s, l)$.
   \squishlisttwoend
  \item {$\mathsf{Appeal}$:} 
  \squishlisttwo
   \item Licensee being accused of leaking $D_{l_x}$ can send a tuple $(\sig_{\ast}(R)$, $r$, $l)$ to contract $\mathcal{C}$ (i.e. judge), where $\sig_\ast(R)$ is signed by both Owner and Licensee. 
   \item $\mathcal{C}$ verifies if $P_l - r \cdot G=R$ and $l \neq l_x$. If yes, Licensee is falsely accused. Otherwise the appeal fails.
   \squishlisttwoend
\squishlistthreeoneend
\squishlistend
}
}
\end{center}
\vspace{-10pt}
\caption{Licensee gets one data $D_l$ from $N$ data $\set{D_1,\ldots,D_N}$ from Owner via $O(1)$-Appeal with $OTEvidence=R$ and $OTRecord=(r,l)$.} \label{fig:ot}
\vspace{-20pt}
\end{figure}

Different from previous work that utilizes transferred vectors (e.g. $\set{E_1,\ldots,E_N}$) in $\mathsf{TransferData}$ for dispute resolving, we find that utilizing one transferred variable $R$ in $\mathsf{GenerateEvidence}$ can be of the same effect. Our key discovery is that $R$ has a one-to-one correspondence to the licensee's chosen index $l$ (see Theorem \ref{th:7}). Therefore, $R$ can be used as the evidence to indicate the licensee's chosen index in the dispute-resolving stage (i.e. the appeal stage). 

As shown in $\mathsf{Appeal}$ in Figure \ref{fig:ot}, if $R$ is signed by owner/licensee and indicates that the corresponding index $l$ differs from accused index $l_x$, we can conclude that the licensee is wrongly accused. We prove the Obliviousness and Non-repudiation of OT with $O(1)$-Appeal in Appendix \ref{sec:appealsecurity}.

Though $O(1)$-Appeal has greatly reduced the on-chain overhead of the appeal stage, there is still a considerable cost of the off-chain bandwidth in $\mathsf{TransferData}$. To reduce the off-chain bandwidth overhead, in Section \ref{sec:construct}, we will further leverage a PIR protocol \cite{xpir} and slightly adapt $O(1)$-Appeal, which guarantees that the size of the data transferred is about the size of data $D_l$. In addition, we will show how to integrate $O(1)$-Appeal with the information-hiding report scheme in Appendix \ref{sec:bigpic}.

\section{Implementing the Argus System}\label{sec:construct}

The previous three sections explain the main objectives and the core ideas of Argus. It is important to recognize that the objectives are not separate problems to solve individually. The Argus contract needs to achieve the objectives altogether in a coherent design. Due to the page limit, we only provide a sketch of our construction and implementation here. The holistic view and the implementation details of the Argus system are provided in Appendix \ref{sec:bigpic} and Appendix \ref{sec:imple}, respectively.

With corresponding watermark algorithms, current Argus system supports three data types: image \cite{image}, audio \cite{audio} and software \cite{software}. A Merkle tree structure is leveraged to reduce the on-chain storage: for any list of data, only the Merkle root of the list is uploaded to the blockchain. We also leverage Private Information Retrieval (PIR) \cite{xpir} to reduce the bandwidth overhead of downloading data for the licensees. 

\section{Security Analysis and Performance Evaluation}\label{sec:eval}
In this section, we first analyze the security of Argus then describe the experimental setup for the performance evaluation. The evaluation results include the performance measurements and the cost of Argus transactions.

\subsection{Security Analysis} \label{sec:analysis}

\begin{table}[]
\centering
\footnotesize
\caption{Interests/threats of participants in argus}
\setlength{\tabcolsep}{3pt}
 \renewcommand{\arraystretch}{1.1}
\begin{tabular}{|c|c|c|}
\hline
Participants & Interest if honest                                                                                         & Threat if malicious                                                             \\ \hline
Owner        & \begin{tabular}[c]{@{}c@{}}To discover infringers and\\ tally the number of copies\end{tabular} & \begin{tabular}[c]{@{}c@{}}To falsely accuse \\ innocent licensees\end{tabular}       \\ \hline
Licensee     & \begin{tabular}[c]{@{}c@{}}To win in an appeal \\ because of innocence\end{tabular}  &  \begin{tabular}[c]{@{}c@{}}To appeal despite \\ the guilt\end{tabular}      \\ \hline
Informer     & \begin{tabular}[c]{@{}c@{}} To submit a honest \\ report once  \end{tabular}   & \begin{tabular}[c]{@{}c@{}}To submit a fake report, \\ steal a report  or submit a \\ valid report multiple times. \end{tabular} \\ \hline
\end{tabular}
\label{tab:interest}
\vspace{-15pt}
\end{table}

The detailed security analysis is given in Appendix \ref{sec:proof}. Without loss of generality, we only present an analysis which considers a game of five participants $\mathsf{Game}^{Argus}_{O, L_1, {L_2},I_1, {I_2}}$, where $O, L_1, {L_2},I_1, {I_2}$, representing the Owner, one Licensee, another Licensee, one Informer and another Informer, respectively. Their interests if honest and their threats if malicious are summarized in Table
\ref{tab:interest}. Based on $\mathsf{Game}^{Argus}_{O, L_1, {L_2},I_1, {I_2}}$, we can easily extend our security analysis to scenarios with multiple owners, informers and licensees. Since $L_1$ and $I_1$ are identical to $L_2$ and ${I_2}$ respectively, to demonstrate the security of Argus, we enumerate all cases that $O, L_1, I_1$ is individually honest (i.e. following the protocol) while other four participants may collude. For each case, we find that the interest of the honest participant will not be affected. In other words, we conclude that {\em if a participant (i.e. owner, licensee or informer) has no fault, the interest of this participant will not be hurt even when others collude.} 

\subsection{Performance Evaluation}
\noindent{\bf Experimental Setup.} Our testbed consists of relatively low-end Azure Virtual Machines (D2s\_v3, 2 vCPUs, 8GB RAM, Linux) for the nodes of owner, licensees, informers and blockchain nodes. For blockchain nodes, we adopt the default PoW algorithm (ie. Ethash \cite{ethereum}) and parameters (block interval, block Gaslimit, etc.) of current Ethereum (date: 2021-04-05) to simulate the public blockchain. The average block interval of Ethereum is set to 12 seconds. The bandwidth of uploading and downloading is tested as around $50$ MB/s. To guarantee a sufficient OT security and a short confirmation of reports, we by default set the OT versions (i.e., N) in Section \ref{sec:detect} as 10,000 while the number of periods (i.e., K) in Section \ref{sec:zkw} as 1000. In other words, the probability to accuse an innocent licensee is $\frac{1}{10000}$ while the report confirmation time during a 180-day period is 4.3 hours. The implementation details of Argus is described in Appendix \ref{sec:bigpic} and Appendix \ref{sec:imple}. 




\vspace{5pt}
\noindent {\bf  Evaluation Results.} We evaluate Argus system from various perspectives of practicality, such as the throughput of system, client latency, gas cost on Ethereum, etc. All protocols of Argus are tested end-to-end. We also give the comparison between Argus and previous work from aforementioned perspectives. Due to the page limit, we only present evaluation results of $\mathsf{Initiate}$, $\mathsf{ShareData}$, $\mathsf{ReportPiracy}$ and $\mathsf{Appeal}$, which are the most resource-consuming protocols (details in Appendix \ref{sec:view}). In other words, these four protocols can introduce considerable overhead in throughput, latency, storage and cost to impede the Argus's adoption in practice. For an intuitive understanding of the gas cost in Ethereum, we represent the gas cost in the number of sending simplest transactions\footnote{Ethereum's simplest transaction only transfers ether and costs 21,000 gas.}.

We first evaluate $\mathsf{Initiate}$, which is the setup phase of Argus. The main elements of $\mathsf{Initiate}$ is to generate an Merkle tree and to deploy the contract. For every licensee, there is a time complexity of $O(N*K)$ for the owner to generate the Merkle tree, which needs about $10$ minutes. This process is totally offline and can be parallelized and accelerated with high-end machines. Deploying the contract costs about $5.2\times 10^6$ gas, which equals to the cost of sending $\sim$248 simplest Ethereum transactions. In addition, we also evaluate the off-chain storage cost for Argus, which is about 960 KB per licensee, which has a space complexity of $O(N)$. 

\begin{table*}[tbhp]
\footnotesize
 \centering
 \caption{High-level comparison between state-of-the-art work and Argus}
 \vspace{-5pt}
 \setlength{\tabcolsep}{3pt}
 \renewcommand{\arraystretch}{1.1}
  \begin{threeparttable}
 \begin{tabular}{llcccccccc}
 \hline
 \multicolumn{1}{c}{Desired properties}            & \multicolumn{1}{c}{Details}           & BSA \cite{bsa} & Custos \cite{custos} & AWM \cite{ot3} & ZKP \cite{zerocash} & Hydra \cite{hydra} & Arbitrum \cite{arbitrum} & This work        \\ \hline
Trusted payments & Full transparency         & $\times$\tnote{*}             & $\surd$                  & N/A\tnote{**}                         & $\surd$                    & $\surd$               & $\surd$                      & $\surd$          \\
   Better payments                   & Timely/guaranteed payout   & $\times$          & $\surd$                  & N/A                         & $\surd$                    & $\times$                & $\times$                      & $\surd$          \\  
   Identifying infringers                  & Strong accusation  & $\surd$         & $\times$                  & $\surd$                     & N/A                & N/A                 & N/A                      & $\surd$     \\ 

 \multirow{2}{*}{Assessing severity}  
 & Sybil-proofness              & $\surd$               & $\times$                      & N/A                         & N/A                    & N/A                & $\surd$              & $\surd$          \\  
 & Information-hiding               & $\surd$               & $\times$                      & N/A                         & $\surd$                    & $\surd$        & $\surd$                       & $\surd$          \\
                                 
 Scalability                     & High throughput       & $\surd$           & $\times$                  & $\times$                     & $\times$                & $\surd$            & $\surd$                  & $\surd$ \\ \hline
 \end{tabular}
 \begin{tablenotes}
 \item[*] Symbols of ``$\surd$'' and ``$\times$'' denote corresponding property is ``achievable'' and ``hard to achieve'', respectively. 
 \item[**] Properties are marked as ``not applicable (N/A)'' if corresponding work is not designed for these properties.
 \end{tablenotes}
 \end{threeparttable}\label{tab:compare}
 \vspace{-15pt}
 \end{table*}
 
\begin{figure}[t]
    \centering
    \includegraphics[width=0.49\textwidth]{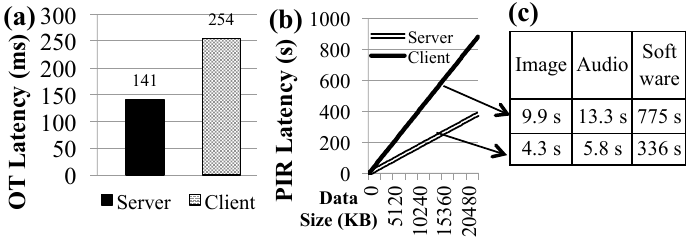}
    \vspace{-15pt}
    \caption{The performance of OT and PIR: (a) OT Latency (b) PIR Latency (c) PIR Latency for different data}
    \label{fig:otpir}
    \vspace{-15pt}
\end{figure}

We evaluate the latency/bandwidth in data sharing. $\mathsf{ShareData}$ includes two phases, OT (Section \ref{sec:detect}) and PIR (Appendix \ref{sec:imple}). While the PIR phase can be done offline, the OT phase can directly affect the throughput of the Argus system especially when there are a number of clients (licensees) concurrently communicating with the server (owner):
\squishlistoneone
 \item For the OT phase, the evaluation result is shown in Figure \ref{fig:otpir} (a). It takes about $141$ ms for the server to complete the phase. Therefore, a throughput of $7.1$ OT requests per second per machine can be served. With a stronger server (Azure D32s\_v3, $32$ cores, $128$GB RAM), the throughput can be further improved to $82.6$ per second. Note that, the throughput can be linearly scaled up by increasing the number of machines.
 \item For the PIR phase, we list the latency of client/server in Figure \ref{fig:otpir} (b) and Figure \ref{fig:otpir} (c). As shown in Figure \ref{fig:otpir} (b), the PIR latency is proportional to the data size, which corresponds to a bandwidth of $206$ Kbps and $476$ Kbps for the client and server, respectively. And with Figure \ref{fig:otpir} (b), we have the PIR latency for data types in Table \ref{tab:wm} (Appendix \ref{sec:wmimple}), which is shown in Figure \ref{fig:otpir} (c). Without PIR, the direct downloading time of the 10,000 copies for the licensee would be 51 s, 67 s and 3800 s, respectively. Similar to the OT phase, the latency and bandwidth of PIR phase can also be linearly and considerably improved with more and higher-performance (i.e. faster core and larger RAM) machines. Since PIR can be offline after the OT phase, the bandwidth of Argus system is determined by the OT phase.
\squishlistoneoneend

\begin{figure}[t]
    \centering
    \includegraphics[width=0.49\textwidth]{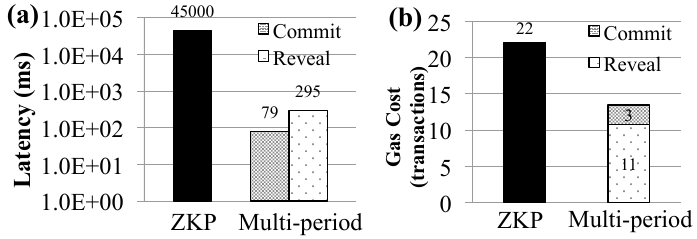}
    \vspace{-20pt}
    \caption{Comparison of previous work (i.e. ZKP \cite{zkbounty,zkpcommit}  ) and multi-period commitment scheme: (a) Latency (b) Gas Cost}
    \label{fig:zkp}
    \vspace{-10pt}
\end{figure}

\begin{figure}[t]
    \centering
    \includegraphics[width=0.49\textwidth]{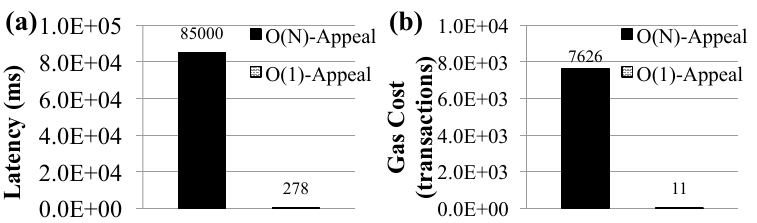}
    \vspace{-20pt}
    \caption{Comparison of previous work (i.e. O(N)-Appeal \cite{ot1,ot2}  ) and O(1)-Appeal: (a) Latency (b) Gas Cost}
    \label{fig:appeal}
    \vspace{-15pt}
\end{figure}

\begin{figure}[t]
    \centering
    \includegraphics[width=0.42\textwidth]{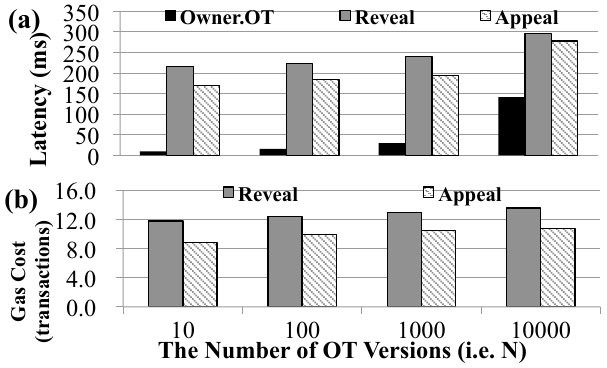}
    \vspace{-5pt}
    \caption{Sensitivity analysis of \#OT versions: (a) Latency (b) Gas Cost}
    \label{fig:sensitivity}
    \vspace{-20pt}
\end{figure}

$\mathsf{ReportPiracy}$ includes the commit phase and the reveal phase (Section \ref{sec:zkw}). The evaluation results of $\mathsf{ReportPiracy}$ are shown in Figure \ref{fig:zkp}. The latency in Figure \ref{fig:zkp} (a) denotes the time which is spent by Informer's machine to submit a transaction to Argus contract until the transaction is executed by blockchain nodes. In other words, the consensus time is excluded in the latency result. The latency of commit and reveal are negligible compared to the block time of Ethereum. 

The on-chain cost of our multi-period scheme is also negligible. As shown in Figure \ref{fig:zkp} (b), the gas cost for commit and reveal are about $8*10^4$ and $2*10^5$, which equals to $\sim$3 and $\sim$11 simplest Ethereum transactions, respectively. In other words, a total cost equivalent to sending 14 simplest transactions is required for an informer to report a piracy in our system. From the gas consumption, given that the maximum gas limit of every Ethereum block is around $12,000,000$, we can conclude that the Transactions Per Second (tps) for ReportPiracy is about $12*10^6/(2.0*10^5*12) \approx 5.0$, which is 11\% of the theoretically maximum Ethereum throughput\footnote{The maximum tps of Ethereum is about $47.6$ when the block only contains simplest Ethereum transactions. By contrast, the average tps of Ethereum is $15.0$ currently (2021-04-05).} (commit and reveal occur in different blocks and thus only the more expensive reveal is considered). We also compare our multi-period scheme with ZKP scheme \cite{zkbounty,zkpcommit}. Compared to ZKP scheme, our scheme can reduce the informer client latency by $99.3$\% and the gas cost by $39$\%.

As in Figure \ref{fig:appeal}, we also evaluate $\mathsf{Appeal}$ (Section \ref{sec:detect}) and get similar results as $\mathsf{ReportPiracy}$: a tps of $4.3$ is achieved given the gas consumption of $\mathsf{Appeal}$ ($\sim 2*10^5$, i.e., the cost of $\sim$11 simplest Ethereum transactions). The appeal protocol proposed by previous work \cite{ot1,ot2} (denoted as ``$O(N)$-Appeal'') introduces unacceptable on-chain operations and exceeds the maximum gas-limit of an Ethereum block. Thus, we cannot evaluate corresponding latency and gas consumption in an end-to-end style. Instead, we estimate them using the number of on-chain operations. Results show that $O(1)$-appeal can significantly reduce the client latency and the gas cost: compared to previous appeal scheme, the informer client latency can be reduced by $99.7$\% and the gas cost by a factor of $960$X, respectively.

As shown in Figure \ref{fig:sensitivity}, we also investigate the impact of choosing different numbers of OT versions (i.e. $N$), from $10$ to $10,000$. On one hand, as addressed in Section \ref{sec:detect}, the value of $\frac{1}{N}$ determines the probability $\phi$ of false accusation; On the other hand, the increase of $N$ has a negative impact on the performance, gas cost and storage overhead of Argus system (Appendix \ref{sec:crypto}). We can see that the increase of OT affects the OT latency of owner significantly (Figure \ref{fig:sensitivity} (a)) while the gas cost increases logarithmically (Figure \ref{fig:sensitivity} (b)) since we use Merkle tree structure in our design (details in Appendix \ref{sec:view})\footnote{We omit the commit operation of informer in Figure \ref {fig:sensitivity} since the latency and gas cost of commit is unrelated to the value of $N$.}.

To summarize the performance of Argus system: in the normal case, which does not involve piracy-reporting or appeal, only $\mathsf{ShareData}$ is involved. In this case, the throughput of Argus system is equal to the throughput of $\mathsf{ShareData}$ ($82.6$ off-chain transactions per second per machine). In the uncommon case of piracy reporting, the throughput of the reporting transactions is $5.0$ tps. In the rare case in which the appeal procedure is performed (i.e., the owner is malicious), the throughput of the appeal transactions is $4.3$ tps.

\section{Related Work}\label{sec:relate}

We summarizes the comparison of Argus with previous work in Table \ref{tab:compare}. Column three and four (i.e. BSA and Custos) are two competitive solutions of Argus while column five to eight (AWM, ZKP, Hydra and Arbitrum) are primitives corresponding to $O(1)$-appeal, multi-period commitment and incentive model, respectively. Details are listed below:
\squishlist
\item{\bf\em Comparison with previous solutions.} Centralized schemes such as BSA \cite{bsa} can considerably disincentivize informers due to the opacity of payments. To increase the trust of payments, Custos \cite{custos} leverages blockchain. However, Custos does not consider strong accusation and cannot assess the severity of piracy, which is important in law enforcement. By contrast, Argus achieves full transparency along with all other properties. 
\item{\bf\em Comparison with state-of-the-art primitives:} 
\squishlistthree
\item Previous work of Asymmetric Watermarking (AWM) \cite{ot1,ot2,ot3,ot4} all rely on the existence of a ``trusted judge''. In addition, their appeal protocols introduce $O(N)$ bandwidth cost, which is unacceptable in blockchain scheme. 
\item Current Zero-Knowledge Proof (ZKP) \cite{zerocash,zkpcommit} can eliminate replay attack by enabling informers to prove their acquaintance of $id_{ij}$ without revealing the answer. However, current ZKP can introduce considerable performance overhead and gas consumption which limits system scalability. 
\item Commitment scheme (or so-called ``commit-reveal'') \cite{arbitrum} can address replay attack by dividing the single submission into phases of commitment and revealing, which is much more efficient than ZKP schemes. However, to fully evaluate piracy, the commitment phase should be relatively long, which delays owner's confirmation of piracy and informers' payments. 
\item  Sybil-proofness \cite{truebit,arbitrum} is introduced to disincentivize informers to report repeatedly: the more times a informer reports, the less bounty the informer can claim. However, existing sybil-proof incentive model merely depends on the total number of informers, which is known only when the collection period ends. 
\item By contrast, Argus overcomes the limitations of above primitives. In addition, further integration and optimization are introduced in this work.
 \squishlistthreeend

\squishlistend

\section{Conclusions}


Anti-piracy is fundamentally a procedure that relies on collecting  data from the open anonymous population, so how to incentivize credible reports is a question at the center of the problem. Academic researchers and real-world companies have come up with various incentive mechanisms. However, without explicitly prescribing the interests of different roles and the objectives of an anti-piracy system, designing such a mechanism has been more of a ``creative art'' than a systematic and disciplined exploration. Currently, there is no good framework to evaluate these designs and actual systems. 

The most essential value of our work is not the Argus system itself, but the approach leading to its design and implementation. We first state clearly the interests of different roles and the goal of full transparency without trusting any role. Once these are stated, all the design requirements naturally surface, such as Sybil-proofness, information-hiding submission, resistance to infringer's repudiation, etc; once these design requirements are clear, we are able to {\em deduce}, rather than {\em invent}, the general form of valid solutions; the deduced general form then boils down to a set of unavoidable technical obstacles, which we overcome by adapting crytopographic schemes, building contract code and optimizing performance. 

Argus exemplifies the outcome of this disciplined approach. It is superior to existing solutions in terms of the trust assumption and the assured properties. In particular, we draw the following conclusions: (1) it is feasible to build a fully transparent solution without introducing a trusted role. This could enable a paradigm shift for anti-piracy incentive solutions. Also, it is a compelling application scenario for public blockchains; (2) such a solution indeed consolidates all roles' interests fairly, i.e., as long as a role is not at fault, his/her interest will not be impaired by other malicious or at-fault roles; (3) besides logic soundness, the solution is economically practical, as a result of our effective optimizations.

\bibliographystyle{IEEEtran}
\bibliography{srds21_argus}
\appendix \label{sec:appendix}

\subsection{Mathematical Deduction of Reward Function} \label{sec:math}

\noindent \textbf{Deducing the General Form with Sybil-proofness.}
Different from \cite{arbitrum}, instead of limiting that $B(I_i,n)=B(I_j,n) (i \neq j)$, we deduce the general form of $B(I_i,n)$ with sybil-proofness merely from the definition. By setting $(m,k,S_m, S_k) =(n-1, n, \{i\}, \set{i,n}), (n-2, n-1, \{i\}, \set{i,n-1}), \ldots, (i, i+1, \{i\}, \set{i,i+1})$ in $ \sum_{i \in S_m} B(I_i,m) \ge \sum_{i \in S_n} B(I_i,n) $, we get :
\begin{eqnarray} \label{eq:inequality}
\left\{ \begin{array}{l}
B(I_i,n) \le  B(I_i,n-1) -  B(I_{n},n)  \\
B(I_i,n-1) \le  B(I_i,n-2) -  B(I_{n-1},n-1) \\
\dots\\
B(I_i, i+1) \le  B(I_i,i) -  B(I_{i+1},i+1)
  \end{array} \right.
\end{eqnarray}

Adding up above equations, we can deduce $B(I_i,n) \le B(I_i,i) -\sum_{j=i+1}^n B(I_j,j)$. By maximizing $B(I_i,n)$ and set the sign of inequality as equality, we have the general form of sybil-proof reward as the following theorem:
\begin{theorem}[General Form of Sybil-proof Reward]\label{th:2}
{\em Let $B(I_i,i)=a_i$ where parameters $\{a_i\} (i=1,2,\dots)$ is a set of arbitrarily positive numbers only if $a_i \ge \sum_{j=i+1}^{\infty} a_j$,  the following equation holds  $\forall n \in \mathbb{N}$ and $\forall i \le n $}:
\begin{equation} \label{eq:form}
 B(I_i,n) = a_i-\sum_{j=i+1}^n a_j
\end{equation}
\end{theorem}

From Theorem \ref{th:2}, we can have following corollary which defines the upper bound of $B(I_i,n)$ and $B(I_i,n)$'s property of exponentially decreasing:
\begin{corollary}[$B(I_i,n)$ decreases exponentially with $i$]\label{co:1}
{\em $B(I_i,n) \le a_1 * 2^{-i+2}$}.
\end{corollary}
\noindent {\em Proof.}
With Equation \eqref{eq:form}, we can have following equation:
$B(I_1,n)+\sum^n_{i=2}B(I_i,n)*2^{i-2} = a_1-2^{n-2}*a_n \le a_1$. Given that $B(I_i,n) \ge 0$, we can conclude that $B(I_i,n)*2^{i-2} \le a_1 \Rightarrow B(I_i,n) \le a_1 * 2^{-i+2}$. \qed

\vspace{5pt}
\noindent \textbf{Enriching with Order-awareness.}
Interestingly, if we set $a_i = c*2^{-i+1}$ in Equation \eqref{eq:form} where $c$ is the maximum reward for informers, we can get $B(I_i,n)=c*2^{-n}$, which is identical to the ones in \cite{arbitrum,truebit}. In addition, from equation \eqref{eq:inequality}, we can find that $B(I_i,n)$ decreases with $n$.

To ensure the property of order-awareness, we apply Equation \eqref{eq:form} to $B(I_i,n) \ge B(I_{i+1},n)$ and have (we set $a_1= c*2^{-1+1} = c$ for simplicity):
\begin{eqnarray}
 \nonumber a_i-\sum_{j=i+1}^n a_j \ge a_{i+1}-\sum_{j=i+2}^n a_j  \Rightarrow a_i \ge 2*a_{i+1} \\
 \Rightarrow  a_i \le c*2^{-i+1} (i \ge 1)
\end{eqnarray}

Therefore, we can set $ a_i = c*2^{-i+1}- \xi_i $ where $\xi_i \in \mathbb{R}_{\ge0}$ and $\xi_1=0$. With $a_i \ge 2*a_{i+1}$ and Theorem \ref{th:2} , we have following theorem about the sybil-proof reward with order-awareness:

\begin{theorem}[Sybil-proofness and Order-awareness]\label{th:3}
{\em With $\set{\xi_i} \in \mathbb{R}^*_{\ge 0}$ where $0 \le \xi_i \le 2* \xi_{i+1} \le c*2^{-i+1}$ and $\xi_1=0$, the following equation holds }:
\begin{equation} \label{eq:order}
 B(I_i,n) = -\xi_i+\sum_{j=i+1}^n \xi_j+c*2^{-n+1}
\end{equation}
\end{theorem}

\noindent \textbf{Enriching with Timely Payout.}
With Equation \eqref{eq:order}, we can divide $ B(I_i, n)$ into two parts defined as follows: ${B_1  \triangleq \lim_{n \rightarrow \infty} B(I_i, n) = -\xi_i+\sum_{j=i+1}^\infty \xi_j}$ and ${B_2  \triangleq B(I_i, n) -B_1(i) = -\sum_{j=n+1}^\infty \xi_j+c*2^{-n+1}}$. Interestingly, we can find that the expression of $B_1$ only relates to the submission index $i$ and $B_2$ only relates to the total submission number $n$, respectively. Thus, we have following theorem about the {\em timely-payout} property:

\begin{theorem}[Sybil-proofness, Order-awareness and Timely Payout]\label{th:4}
{\em $B_1(i)=-\xi_i+\sum_{j=i+1}^\infty \xi_j$ can be allocated to the $i$-th successful Informer $I_i$ {\bf\em immediately}, while $B_2(n)= B(I_i, n) -B_1 = -\sum_{j=n+1}^\infty \xi_j+c*2^{-n+1}$ can be allocated to $I_i$ after $n$ is determined (i.e. after the collection period ends)}.
\end{theorem}

\vspace{5pt}
\noindent \textbf{Enriching with Guaranteed Amount.}
By examining the expression of $B_2(n)$, we can get\footnote{Since $B(I_i, n)$ decreases with $n$, $B(I_i, n)$ always exceeds $B_1(i)=\lim_{n \rightarrow \infty} B(I_i, n)$. Thus, $B_2(n)=B(I_i, n)-B_1(i) \ge 0$.} the inequality of $B_2(i)$: 
\begin{eqnarray}
 \nonumber 0 \le  B_2(n) \le c*2^{-n+1} \Rightarrow \lim_{n \rightarrow \infty}B_2(n) = 0  \\ 
 \Rightarrow \lim_{n \rightarrow \infty} B(I_i,n)= \lim_{n \rightarrow \infty} B_1(i) = B_1(i)
\end{eqnarray}

Therefore, to achieve the property of {\em guaranteed amount}, we should properly set the values of $\set{\xi_i}$ to ensure $\lim_{n \rightarrow \infty}B_1(i)>0\quad (i \le l)$. 

With $\xi_i \le 2* \xi_{i+1}$ in Theorem \ref{th:3}, we introduce set $\set{\Delta_i} \in  \mathbb{R}^{*}_{\ge 0}$ which satisfy $\xi_{i+1}=\frac{1}{2}\xi_i + \Delta_{i}$. Then, from the formula of $B_1(i)$ we have:
\begin{equation} \label{eq:b1}
 B_1(i) = 2*\sum_{j=i}^\infty \Delta_{j}
 \end{equation}
 
 From $\xi_{i+1}=\frac{1}{2}\xi_i + \Delta_{i}$ and $\xi_{1}=0$, we can get:
 \begin{eqnarray} \label{eq:x1}
  \xi_{i+1}=\sum_{j=1}^{i} 2^{j-i}*\Delta_j\\
  \Rightarrow \sum_{j=1}^{i} 2^{j}*\Delta_j = \xi_{i+1} * 2^{i} \le c \label{eq:constraint1}
 \end{eqnarray}
 
 Since $ \sum_{j=1}^{i} 2^{j}*\Delta_j $ increases with $i$, thus Equation \eqref{eq:constraint1} is equivalent to:
 \begin{equation} \label{eq:constraint2}
 \sum_{j=1}^{\infty} 2^{j}*\Delta_j \le c
 \end{equation}
 
Therefore, we have following theorem for reward function achieving four aforementioned properties:

\begin{theorem}[Sybil-proofness, Order-awareness, Timely, \newline Payout and Guaranteed Amount]\label{th:5} 
{\em $\lim_{n \rightarrow \infty} B(I_i, n)>0$ only if $\set{B_1(i)}$ defined in Theorem \ref{th:4} are expressed by Equation \eqref{eq:b1} where non-zero set $\set{\Delta_i} \in  \mathbb{R}^{*}_{\ge 0}$ satisfy Equation \eqref{eq:constraint2}}.
\end{theorem}

With Theorem \ref{th:5}, we can directly have Theorem \ref{th:expr} in Section \ref{sec:quantify}.

\subsection{Security Analysis of $O(1)$-Appeal} \label{sec:appealsecurity}
In this section, we first prove the one-to-one correspondence between $R$ and $l$. Then, we demonstrate the obliviousness and the non-repudiation of $O(1)$-Appeal.
\begin{theorem}[One-to-one Correspondence between $R$ and $l$]\label{th:7}
{\em The transferred message $R$ in $\mathsf{Appeal}$ corresponds to $l$ in $\mathsf{GenerateEvidence}$.}
\end{theorem}
\begin{proof}
Let us assume that one $R$ can correspond to two chosen index $l$ and $l'$ (then the licensee can deny the accusation). Therefore, from $\mathsf{GenerateEvidence}$, the licensee should know a tuple $(l,l',r,r')$ to satisfy $R = P_l- r \cdot G$ and $R = P_{l'}- r' \cdot G$, which can deduce the following equation:
\begin{equation} \label{eq:7}
(r-r') \cdot G = P_{l}-P_{l'}
\end{equation}

However, since $\set{P_i}$ are randomly generated by the owner, according to Discrete Logarithm Assumption, the possibility is negligible that the licensee finds $r$ and $r'$ to fulfill Equation \eqref{eq:7}. Thus, we can conclude that $R$ has a one-to-one correspondence to $l$. 
\end{proof}

With Theorem \ref{th:7}, we can now prove the obliviousness and non-repudiation of $O(1)$-Appeal. Accordingly, we have following two theorems:

\begin{theorem}[Obliviousness of $O(1)$-Appeal]\label{th:8}
{\em With $\mathsf{Initilize}$, $\mathsf{GenerateEvidence}$ and $\mathsf{TransferData}$ in Figure \ref{fig:ot}, the owner can transfer $D_l$ to the licensee with obliviousness.}
\end{theorem}
\begin{proof} 
Since the first three sub-functions of $O(1)$-Appeal (i.e. $\mathsf{Initilize}$, $\mathsf{GenerateEvidence}$ and $\mathsf{TransferData}$) are directly derived from \cite{np}, which is malicious-secure, the two features of Obliviousness are satisfied. 
\end{proof}

Note that after $\mathsf{Appeal}$, the index chosen by the licensee is revealed which seems to violate Obliviousness. However, at that time, the licensee should be in the EXONERATED state (see Appendix \ref{sec:bigpic} for details).

\begin{theorem}[Non-repudiation of $O(1)$-Appeal]\label{th:9}
{With $\mathsf{Appeal}$ in Figure \ref{fig:ot}, both the owner and the licensee cannot deny which copy (i.e. $l$) the licensee chose previously.}
\end{theorem}
\begin{proof}
According to $\mathsf{GenerateEvidence}$, $R$ is signed by both the owner and the licensee. Thus, the licensee cannot fake another $R$ to submit to the contract. Therefore, based on Theorem \ref{th:7}, in order to guarantee that $P_l - r \cdot G=R$ in $\mathsf{Appeal}$, the licensee can only submit the correct $l$ used in $\mathsf{GenerateEvidence}$. In other words, the licensee cannot successfully appeal if accused with a correct $l$. Similarly, since $R$ is signed by the owner and Theorem \ref{th:7} holds, the owner cannot deny the $l$ which can pass $P_l - r \cdot G=R$.  
\end{proof}

\subsection{Construction Details of Argus System} \label{sec:bigpic}

{
\renewcommand{\baselinestretch}{0.3}
\setlength{\baselineskip}{11.4pt}

\begin{figure*}[thpb]
\begin{center}
\fbox{ 
\procedure[syntaxhighlight=false,mode=text,codesize=\small]{}{
\hspace{-8pt} $\mathsf{Initiate}$\\
\squishlist
   	\item INPUTS:
	\squishlistthree
	\item security parameter $(\lambda, N)$.
	\item the number of collection periods $K$.
	\item the number of Licensees $M$.
	\item copyrighted data $D$.
	\squishlistthreeend
	\item OUTPUTS: 
	\squishlistthree
	\item secret value $s$ of Owner.
	\item public parameters $a_{s}$ and $\vec{P}=\set{P_1, \ldots , P_N}$.
	\item $M\times N$ versions of watermarked data $D^w_{M\times N}$.
	\item $\mathsf{IDMap/IDTree}$ locally stored by Owner.
	\squishlistthreeend
\squishlisttwo 
	\item Owner deploys the contract $\mathcal{C}$ to a ledger.
	\item Owner randomly samples a secret $s \gets \bin^\lambda$. 
	\item Owner computes and publishes $a_s \gets \pgen(s)$.
	\item For each $j \in [N]$: 
	\squishlistfour
	\item Owner randomly samples $P_j \gets \pgen(\bin^\lambda)$.
	\item Owner invokes $\mathcal{C}.\mathsf{Store(\textrm{``$p$''},\ast)}$ where $\ast$ as $P_j$.
	\squishlistfourend
	\item For each  $i \in [M]$ and each $j \in [N]$: 
	\squishlistfour
	\item Owner randomly samples $id_{ij} \gets \bin^\lambda$.
	\item Owner watermarks $D$ with $id_{ij}$: $D^w_{ij} \gets E_w(D, id_{ij})$. 
	\item Owner updates a Hashmap:
		$\mathsf{IDMap.Insert}(\mathcal{H}(id_{ij}):(i,j))$.
	\item Owner updates a Merkle tree by inserting $K$ items: \newline
		$\mathsf{IDTree.Insert}(\mathcal{H(H(}id_{ij} || 1)|| i || j),\ldots,\mathcal{H(H(}id_{ij} || K)|| i || j))$.
	\squishlistfourend
	\item Owner invokes $\mathcal{C}.\mathsf{Store(\textrm{``$rt$''},\ast)}$ where $\ast$ as $\mathsf{IDTree.root()}$.
	\item Owner invokes $\mathcal{C}.\mathsf{Deposit()}$ with $v$ coins for each Licensee.
\squishlisttwoend
\squishlistend
\\
\\
\hspace{-8pt} $\mathsf{ShareData}$\\
\squishlist
   	\item INPUTS:
	\squishlistthree
	\item public parameters $\vec{P}$ and $a_{s}$.
	\item secret value $s$ of Owner.
	\item Licencee$_i$'s private key $sk_i^L$ and Owner's private key $sk^O$.
	\item $N$ versions of watermarked data $\overrightarrow{D^w_{i}}=\set{D^w_{i1},\ldots, D^w_{iN}}$.
	\squishlistthreeend
	\item OUTPUTS: 
	\squishlistthree
	\item secret value $r_i$ and $l_i$ of Licensee$_i$.
	\item $\sig_\ast(R_i)$ signed and kept by Owner and Licensee$_i$.
	\item A watermarked data $D^w_{il_{i}} \in \overrightarrow{D^w_{i}}$ sent to Licensee$_i$.
	\squishlistthreeend
\squishlisttwo
	\item Licensee$_i$ randomly samples $r_i \gets \bin^\lambda$ and $l_i \in [N]$.
	\item $  \sig_\ast(R_i) \gets \mathsf{OT.GenerateEvidence}(\vec{P}, r_i, l_i, sk_i^L, sk^O)$ \newline where $\sig_\ast(R_i)$ denotes $R_i$ signed by Owner and Licensee$_i$.
	\item 
	$D^w_{il_{i}} \gets \mathsf{OT.TransferData}(\vec{P}, a_s, s, \sig_{\ast}(R_{i}), r_i, l_i, \overrightarrow{D^w_{i}})$. 
\squishlisttwoend
\squishlistend
\\
\\
\hspace{-8pt} $\mathsf{ReportPiracy}$
\squishlist
   	\item INPUTS:  
	\squishlistthree
	\item a pirated copy $D^w_{xy}$.
	\item Informer$_i$'s address $pk_i^I$.
	\item the current period number $T$.
	\squishlistthreeend
	\item OUTPUTS: bit $b$ equals $\true$ if $\mathcal{C}.\mathsf{VerifyReport()}=\true$.
\squishlisttwo
	\item Informer$_i$ detects watermark in $D^w_{xy}$ and extracts secret ID: \newline
	$ id_{xy} \gets D_w(D^w_{xy})$.
	\item Informer$_i$ queries Owner for following information: 
	\squishlistfour
	\item  $(x,y)\gets\mathsf{IDMap.Query}(\mathcal{H}(id_{xy}))$.
	\item  $\mathsf{Path}_{xyT} \gets \mathsf{IDTree.Query(x,y,\textrm{\em T})}$.
	\squishlistfourend 
	\item Informer$_i$ invokes  $\mathcal{C}.\mathsf{Store(\textrm{``$cm$''},\ast)}$ where $\ast$ is a tuple \newline 
	$(cm_1,cm_2,cm_3)=(\mathcal{H(H(}id_{xy} || T)|| pk_i^I),x,y)$.
	\item In period $T+1$, Informer$_i$ invokes $\mathcal{C}.\mathsf{VerifyReport()}$ with \newline
	a tuple $(rv_1,rv_2,rv_3)=(\mathcal{H(}id_{xy}||T), \mathsf{Path}_{xyT}, pk_i^I)$.
\squishlisttwoend
\squishlistend
\\
\\
\hspace{-8pt} Omitted protocols$^*$: $\mathsf{Appeal}$, $\mathsf{ClaimBounty}$, $\mathsf{ConfirmInfringer}$
}
\pchspace[0cm]
\procedure[syntaxhighlight=false,mode=text,width=0.44\textwidth,codesize=\small]{}{
\vspace{4pt}
\hspace{-8pt} $\mathcal{C}.\mathsf{Store()}$
\squishlist
   	\item INPUTS: data type $tp$, data $dat$.
	\item OUTPUTS: bit $b$ equals $\true$ if invocation succeeds. 
\squishlistend
\squishlistfive
	\item If $tp=\textrm{``$p$''}$, $\mathcal{C}.\mathsf{PList.Append(dat)}$.
	\item If $tp=\textrm{``$rt$''}$, $\mathcal{C}.\mathsf{rt}=\mathsf{dat}$.
	\item If $tp=\textrm{``$cm$''}$, $\mathcal{C}.\mathsf{CMList_\textrm{\em T}.Append(dat)}$ ($T=\mathcal{C}.\mathsf{Time()}$).
\squishlistfiveend
\\
\\
\hspace{-8pt} $\mathcal{C}.\mathsf{VerifyReport()}$
\squishlist
   	\item INPUTS: $rv_1,rv_2,rv_3$
	\item OUTPUTS: bit $b$.
\squishlistend
\squishlistfive
	\item $T=\mathcal{C}.\mathsf{Time()}$.
	\item Output $b=\true$ if all following conditions are $\true$:
	\squishlistfour 
	 \item If $rv_2$ is a valid Merkle tree path with root $\mathcal{C}.\mathsf{rt}$. 
	\item  If $\mathcal{H}(rv_1||rv_3)=cm_1$ where $cm_1$ is the first element \newline 
	in a tuple $(cm_1, x, y) \in\mathcal{C}.\mathsf{CMList_{T-1}}$.
	\item If $rv_2$ contains $\mathcal{H(H}(rv_1)||x||y)$.
	\squishlistfourend
	\item If $b=\true$:
	\squishlistfour 
	\item If $\mathcal{C}.\mathsf{Status}_x=$NORMAL: 
	\squishlistsix
	\item $\mathcal{C}.\mathsf{ReportTime_{x}} \gets T-1$.  
	\item $\mathcal{C}.\mathsf{Version}_x \gets y$. 
	\item $\mathcal{C}.\mathsf{Status}_x \gets $ACCUSED.
	\squishlistsixend
	\item $\mathcal{C}.\mathsf{ReportNumber_{x}}\gets \mathcal{C}.\mathsf{ReportNumber_{x}}+1$.
	\item $\mathcal{C}.\mathsf{IsInformer}_{x}[rv_3]=\true$.
	\item $\mathcal{C}.\mathsf{SendBounty}(rv_3$, $\mathcal{B_\mathrm{1}(C}.\mathsf{ReportNumber_{x}},v)$).
	\squishlistfourend
\squishlistfiveend
\\
\\
\hspace{-8pt} $\mathcal{C}.\mathsf{VerifyAppeal()}$
\squishlist
   	\item INPUTS: $r_x$, $l_x$, $\sig_\ast(R_{x})$
	\item OUTPUTS: bit $b$.
\squishlistend
\squishlistfive
	\item Output $b=\true$ if all following conditions are $\true$:
	\squishlistfour 
	\item  If $\sig_\ast(R_{x})$ is signed by Owner and Licensee$_x$.
	\item  If $\mathcal{C}.\mathsf{PList}[l_x]-\pgen(r_x)=R_x$.
	\item If $\mathcal{C}.\mathsf{Status}_x = $ ACCUSED and $l_x \neq \mathcal{C}.\mathsf{Version}_x$.
	\item  If $\mathcal{C}.\mathsf{Time()}-\mathcal{C}.\mathsf{ReportTime_{x}} \le Timeout$.
	\squishlistfourend
	\item If $b=\true$: $\mathcal{C}.\mathsf{Status}_x \gets $EXONERATED.
\squishlistfiveend
\\
\\
\hspace{-8pt} $\mathcal{C}.\mathsf{AllocateBounty()}$
\squishlist
   	\item INPUTS: Informer's address $pk^I_i$ and Infringer's index $x$.
	\item OUTPUTS: bit $b$.
\squishlistend
\squishlistfive
	\item Output $b=\true$ if both following conditions are $\true$:
	\squishlistfour 
	\item  If $\mathcal{C}.\mathsf{IsInformer}_{x}[pk^I_i] = \true$.
	\item  If $\mathcal{C}.\mathsf{Time()} \ge K$.
	\squishlistfourend
	\item If $b=\true$: 
	\squishlistfour
	\item $\mathcal{C}.\mathsf{SendBounty}(pk^I_i$, $\mathcal{B_\mathrm{2}(C}.\mathsf{ReportNumber_{x}},v)$).
	\item $\mathcal{C}.\mathsf{IsInformer}_{x}[pk^I_i]=\false$.
	\squishlistfourend
\squishlistfiveend
\\
\\
\hspace{-8pt} $\mathcal{C}.\mathsf{SetGuilty()}$
\squishlist
   	\item INPUTS: Infringer's index $x$.
	\item OUTPUTS: bit $b$.
\squishlistend
\squishlistfive
	\item Output $b=\true$ if both following conditions are $\true$:
	\squishlistfour 
	\item  If $\mathcal{C}.\mathsf{Status}_x = $ ACCUSED.
	\item  If $\mathcal{C}.\mathsf{Time()}-\mathcal{C}.\mathsf{ReportTime_{x}} > Timeout$.
	\squishlistfourend
	\item If $b=\true$: $\mathcal{C}.\mathsf{Status}_x = $ GUILTY.
\squishlistfiveend
\\
\\
$^*$\small{For simplicity, protocols for $\mathsf{Appeal}$, $\mathsf{ClaimBounty}$ and $\mathsf{ConfirmInfringer}$ are omitted in the figure which invoke $\mathsf{VerifyAppeal()}$, $\mathsf{AllocateBounty()}$ and $\mathsf{SetGuilty()}$, respectively. These protocols can be directly inferred from corresponding contract functions.}
} 
}
\end{center}
\vspace{-10pt}
\caption{Construction of an anti-piracy scheme $(\mathsf{Initiate, ShareData, ReportPiracy, Appeal, ClaimBounty, ConfirmInfringer})$ with Argus contract $\mathcal{C}=(\mathsf{Store(), VerifyReport(), VerifyAppeal(), AllocateBounty(), SetGuilty()})$.}
\label{fig:construct}
\end{figure*}
}

\subsubsection{Cryptographic building blocks} \label{sec:buildblock}
Let $\lambda$ denote the security parameter and $q$ as a large prime number. The building blocks used in our construction are as follows:
\squishlistoneone
\item{\bf Pseudorandom Number Generator} is a function that outputs a random $\lambda$-bit string $\bin^{\lambda}$.
\item{\bf ECC Point Mapping $\pgen$} is a function that maps a field number over $\ZZ_q$ via the multiplication of an ECC generator $G$ to a point belongs to group $\mathbb{G}$. 
\item{\bf Collision Resistant Hashing $\mathcal{H}$} is a function that $\mathcal{H}:\bin^* \to \bin^{O(\lambda)}$ which maps an arbitrary length of bit string to a $O(\lambda)$-bit random string. $\mathcal{H}$ satisfies the properties of {\em collision-resistance} and {\em irreversibility}.
\item{\bf Merkle Tree $\mathcal{MT}$} is a tree-based data structure where every father node $Node_i$ stores the value of the hash value of string catenation of $Node_i$'s two child nodes $Node_{2i}$ and $Node_{2i+1}$: $\mathcal{H}(V_{2i} || V_{2i+1}) \to V_{i}$. With {\em collision-resistance} of $\mathcal{H}$, any changes in leaf nodes leads to a unique $\mathcal{MT}$ tree root.
\item{\bf Oblivious Transfer $\mathsf{OT}$} is a function described in Figure \ref{fig:ot} which mainly comprises two sub protocols: $\mathsf{OT.GenerateEvidence}$ and $\mathsf{OT.TransferData}$.  $\mathsf{OT.GenerateEvidence}$ can generate the evidence of licensee's choice which is later used in $\mathsf{OT.TransferData}$. In addition, the evidence functions in the phase of licensee's potential appeal. And for appeal, we leverage $O(1)$-appeal in Section \ref{sec:detect}.
\item{\bf Reward Function $\mathcal{B}_1$} and $\mathcal{B}_2$ are two functions described in Section \ref{sec:quantify} which can incentivize informers' good-faith reports with some other advantageous properties. $\mathcal{B}_1$ is the value of reward allocated to informers immediately which only relates the order of informers' submission while $\mathcal{B}_2$ is the value of reward after deadline which only relates to the total number of submissions.
\item {\bf Watermarking Functions} $E_w$ and $D_w$ are watermark embedding and detection functions, respectively. Via $E_w$, owner can add a piece of information into certain types of data imperceptibly; and via $D_w$, anyone can robustly detect the information embedded into the watermarked data unless unacceptable distortion is applied to the data. 
\squishlistoneoneend

\subsubsection{A detailed description of Argus system} \label{sec:view}
Figure \ref{fig:construct} gives a detailed description about the Argus system, 
which consists of six protocols: $\mathsf{Initiate}$, $\mathsf{ShareData}$, $ \mathsf{ReportPiracy}$, $ \mathsf{Appeal}$, $\mathsf{ClaimBounty}$ and $\mathsf{ConfirmInfringer}$. These protocols execute between \em{client programs} representing different roles and the \em{contract code} running on the public blockchain. For simplicity, we omit the address-registration protocol, in which the owner and the licensees associate their identities with key pair $(pk,sk)$\footnote{The linkability can be only visible to the authority for future law enforcement. In addition, we do not require informers to register their identities.}. The contract $\mathcal{C}$, shown in the figure, includes five functions: $\mathsf{Store()}$, $\mathsf{ VerifyReport()}$, $\mathsf{ VerifyAppeal()}$, $\mathsf{AllocateBounty()}$ and $\mathsf{SetGuilty()}$\footnote{Protocols in Figure \ref{fig:construct} also invoke $\mathcal{C}.\mathsf{Deposit()}$ and $\mathcal{C}.\mathsf{SendBounty()}$, which can be constructed in a straightforward way thus omitted in the figure.}.

$\mathsf{Initiate}$ establishes the basis for the other protocols and contract functions. In $\mathsf{Initiate}$, parameters of OT (Step 1 to 3), watermarked data (Step 4.a, 4.b) and targets/bounty for piracy report (Step 5.c, 5.d, 6, 7) are prepared. Note that in Step 5.d, the mechanism of information-hiding report (Section \ref{sec:zkw}) is involved. A Merkle tree $rt$ is introduced to reduce the on-chain storage overhead of OTEvidence, which is proportional to $O(N*M)$. Similarly, $\vec{P}$ can be also stored on-chain via the Merkle tree proof to avoid the overhead of $O(N)$.

$\mathsf{ShareData}$ is the protocol for licensees to obtain data. As described in Section \ref{sec:detect}, the core building block within the protocol is  OT adapted for $O(1)$-Appeal. 

$\mathsf{ReportPiracy}$ is the protocol for the open population to participate in the anti-piracy solution. Thus, we should make this process as cost-efficient and incentive-compatible as possible. The core of this process is the multi-period commitment scheme (Step 3 and 4), which is described in Section \ref{sec:zkw}. In addition, the informer is supposed to fetch some assistant data  (i.e. the Merkle tree path) to generate a transaction to report piracy via $\mathcal{C}.\mathsf{VerifyReport()}$.

$\mathsf{Appeal}$ can be regarded as the subsequent protocol of $\mathsf{ShareData}$ in case that the innocent Licensee is incriminated by malicious owner in $\mathsf{ReportPiracy}$. From Section \ref{sec:detect}, we can see that only $O(1)$ messages are sent to the contract to invoke $\mathcal{C}.\mathsf{VerifyAppeal()}$ before timeout.
 
 $\mathsf{ClaimBounty}$ is the protocol for informers to claim the rest of bounty (i.e. $\mathcal{B}_2(n)$ in Section \ref{sec:quantify}) after collection period ends via $\mathsf{AllocateBounty()}$. Note that, the first part of bounty $\mathcal{B}_1(i)$ has allocated upon informers' reveal phase ($\mathcal{C}.\mathsf{VerifyReport()}$).

 $\mathsf{ConfirmInfringer}$ is the protocol for owner to set the status of accused licensee into ``GUILTY'' if the licensee does not successfully appeal during timeout.

Corresponding contract functions are also detailed in Figure \ref{fig:construct} and referred to in aforementioned descriptions about protocols.

We further optimize the performance, gas consumption, storage of Argus design, which is introduced in Appendix \ref{sec:imple}. 


 \subsection{Implementation details} \label{sec:imple}

Argus' implementation\footnote{We will publish the source code of our implementation in the near future.} spans over three technology areas: watermarking (embedding, detection), cryptography (OT, PIR, etc.), contract (contract code and client script). For watermarking, we integrate existing algorithms of three data types (i.e., image, audio, software) currently. Other data types could be added in the future. For cryptography and contract, we implement and optimize the designs using the insights and ideas described in Appendix \ref{sec:contract} and Appendix \ref{sec:crypto}. 

Instantiation is about the setup of parameters and algorithms described in Appendix \ref{sec:bigpic}, which directly impacts the security analysis of Argus in Section \ref{sec:eval}. We will provide the parameter and the algorithm utilized in this work especially in Appendix \ref{sec:crypto}.

\subsubsection{Digital Watermarking}\label{sec:wmimple}
We leverage Spread Spectrum (SS) based watermark for the images \cite{image} and audio \cite{audio}; Control flow graph (CFG) watermark for the software \cite{software}. As described in Section \ref{sec:model}, watermark-related attack \cite{sensitivity, bibd} is not considered in this work. Corresponding countermeasures are orthogonal to this work and thus can be further applied to increase the security of our system. The setup of watermarking parameters are listed in Table \ref{tab:wm}. In addition, we also list the performance and effect of watermark embedding/detection. 

\vspace{5pt}
\noindent {\bf Further Optimization.} Intuitively, the watermark embedding may take considerable time for the owner to initiate if the $N$ is large with an overhead of $O(N*M)$. In fact, this overhead can be mitigated with segment-and-watermark technology \cite{ot2,ot1}: to split the data into $L$ multiple segments and randomly embed each segment with an alternative watermark from $\set{w_1,w_2}$. Then, by randomly combining the segments, it is easy to obtain $2^L$ versions of watermarked data. The whole process only costs twice of watermark embedding time. In addition, since the watermark embedding time is offline and setup only once, the overhead of this part is totally acceptable.

\begin{table}[thbp]
\centering
\footnotesize
\caption{The details of watermarking schemes}
\vspace{-5pt}
\setlength{\tabcolsep}{3pt}
 \renewcommand{\arraystretch}{1.1}
 \begin{threeparttable}
\begin{tabular}{|c|c|c|c|}
\hline
\multirow{2}{*}{Types}   & \multirow{2}{*}{Configuration \& Implementation}                                                                            & \multicolumn{2}{c|}{Latency} \\ \cline{3-4} 
                         &                                                                                                                         & $E_w$     & $D_w$    \\ \hline
Image  \cite{image}      & \begin{tabular}[c]{@{}c@{}}512$\times$512 image ($\sim 256$ KB), \\ 128 segments, 40dB PSNR, DCT algorithm\end{tabular} & 3.0 s         & 2.2 s        \\ \hline
Audio \cite{audio}       & \begin{tabular}[c]{@{}c@{}}44100 Hz, 1 second wav ($\sim$ 344 KB), \\ MCLT-based watermarking algorithm\end{tabular}    & 4.8 s         & 3.1 s        \\ \hline
Software \cite{software} & \begin{tabular}[c]{@{}c@{}}Geth binary ($\sim$ 20 MB), \\ Obfuscated CFG with 128-bit watermark\end{tabular}            & 1.7 s         & 2.5 s        \\ \hline
\end{tabular}
\end{threeparttable} \label{tab:wm}
\vspace{-10pt}
\end{table} 

\subsubsection{Contract} \label{sec:contract}
For the ledger, we leverage one of the most popular and commercially successful platforms, Ethereum \cite{ethereum}, to deploy and test our contract. Accordingly, we implement Argus contract with $\sim 2,400$ lines of code in Solidity $^\land0.4.23$ \cite{solidity}. It is developed in Truffle $4.1.14$ \cite{truffle} and deployed to a blockchain of $10$ nodes running 
Geth istanbul version v$1.9.3$ \cite{geth} to emulate the public Ethereum (similar to Ropsten Testnet\cite{testnet}).   

\vspace{5pt}
\noindent {\bf Further Optimization.} We further optimize $\mathsf{VerifyReport}$, which is one of the most expensive functions of contract $\mathcal{C}$. We apply a caching strategy, of which the main idea is simple: every time every informer submits a path of Merkle tree to the contract for verification, which not only introduces considerable on-chain bandwidth overhead but also on-chain hash operations. To mitigate the overhead, we let the contract memorize the path submitted by previous informers. For the incoming informers, only the non-overlapped part of the path is submitted and checked by the contract $\mathcal{C}$. Therefore, given the temporal locality of piracy distribution, most of the paths submitted by adjacent informers are overlapped and can be omitted. Our results show that the caching strategy can reduce 31\% of the gas consumption.

\subsubsection{Cryptography} \label{sec:crypto}
For cryptographic modules, we customize our OT protocol\cite{np} based on libOTe \cite{libote}, which is one of the state-of-the-art implementations of OT atop the Miracle Library \cite{miracle}. We set $N$ as $10,000$ to achieve an sufficient security strength. For the symmetric encryption, asymmetric encryption (e.g. for generating signature) and hash algorithms, we set $\lambda$ as $128$ and leverage AES, ECCSA, SHA-256 implemented in the Miracle library, respectively. In addition, in order to be compatible with Web3 in Ethereum \cite{web3}, we adapt the ECC in Miracle to be based on BN-$254$ curve \cite{bn254}. And with the SHA-256 algorithm, we implement the Merkle tree algorithm. For multi-period commitment scheme, we set $K$ as $1000$ and the period of report collection as $180$ days, which means that an informer's report is confirmed at most in $180*24/1000=4.32$ hours.  In addition, we leverage state-of-the-art PIR, XPIR, which is highly optimized by using Ring-LWE \cite{xpir}.

\vspace{5pt}
\noindent {\bf Further Optimization.} We optimize $\mathsf{Initiate}$ to reduce storage overhead for the owner. For a straightforward design, we can have the Merkle tree as a three-layered structure: the first layer (near the tree top $rt$) has $O(M)$ leaves which corresponds to licensees; the second layer, of which the tree top is the leave of the first layer, has $O(M*N)$ leaves which corresponds to $N$ versions of OT data; and similarly, the third has $O(M*N*K)$ corresponding to the timestamps. For the owner, only the leaves of the second layers are stored locally and provided with a storage overhead of $O(N)$ per licensee. With these leaves, the informers can generate the entire Merkle tree according to the timestamp and these leaves. 


We also optimize the off-chain bandwidth overhead incurred in $\mathsf{TransferData}$. For OT protocol where $N=10,000$, compared to the intended data, $10,000$ times of data size are transferred between the owner and the licensee (Figure \ref{fig:ot}). In addition, the performance overhead of asymmetric encryption/decryption also burdens the owner/licensee. Therefore, we adapt the process of OT to mitigate the overhead of asymmetric encryption and introduce Private Information Retrieval~(PIR) \cite{xpir} to mitigate the bandwidth overhead: data are first encrypted with symmetric encryption (e.g. AES) and only keys are transferred via OT. The encrypted data are then downloaded by the licensee with PIR protocol in an offline style. With these optimizations, the overhead of bandwidth and performance can be largely reduced, which is shown in Section \ref{sec:eval}.
%

\subsection{Security Analysis of $\mathsf{Game}^{Argus}_{O, L_1, L_2,I_1, I_2}$} \label{sec:proof}

In this section, we enumerate cases that each participant (i.e. $O$, $L_1$, and $I_1$) is honest and demonstrate that the interest of the participant will not be affected even if the rest of participants collude according to their interests. The security model (i.e. interests of participants) is introduced in Section \ref{sec:analysis}.

 \begin{proposition} \label{pro:1}
 If Owner $O$ is honest, the interest of $O$ is guaranteed in $\mathsf{Game}^{Argus}_{O, L_1, L_2,I_1, I_2}$. 
 \end{proposition} 
 \begin{proof}
There are eight cases as follows except those identical:
 \begin{itemize}
  \item {\em Malicious $L_1$, Honest $I_1$.} As soon as the licensee $L_1$ redistributes the pirated copies to informers, informer $I_1$ can detect the watermark and extract the secret ID to report. Moreover, the licensee cannot invalidate the report according to $O(1)$-Appeal protocol (Section \ref{sec:detect}). Therefore, the infringer is correctly identified and the (lower bound of) number of pirated copies is indicated (unless $L_1$ over-report himself/herself for redistributing piracy, which conflicts $L_1$'s interests). 
   \item {\em Honest $L_1$, Malicious $I_1$.} If $L_1$ does not redistribute the data, the possibility for $I_1$ to guess any secret ID $id_{xy}$ is $O(2^{-\lambda})$ which is negligible. In other words, $I_1$ cannot mislead the Owner to accuse innocent $L_1$.
     \item {\em Malicious $L_1$, Malicious $I_1$.} The analysis is similar to above two cases. First, $L_1$ should be correctly identified. Second, due to the property of Sybil-proofness, $I_1$ is disincentivised to over-submit. Thus, the number of pirated copies indicated by submissions is correct.
      \item {\em Malicious coalition $(L_1, I_1)$.}  Since $L_1$ is financially motivated, $L_1$ may compensate $I_1$ with other rewards to be not accused. However, due to the nature of open population that $L_1$ cannot collude with every potential informers, another informer (e.g. $I_2$) may report to Argus contract $\mathcal{C}$, which invalidates the collusion of $L_1$ and $I_1$.
 \item {\em Malicious coalition $(L_1, L_2)$.} Since the secret ID $id_{ij}$ are unique for every licensee and the bounty for every licensee is individually configured, the anti-piracy process can be regarded as orthogonal to every licensee. Thus, the collusion has no effects on the interests of owner. The only effect is that $L_1, L_2$ may collude to launch collusion attack to bypass the watermark detection, which can be mitigated by collusion-resistant code \cite{tardos}.
  \item {\em Malicious coalition $(I_1, I_2)$.} In fact, since we do not hold any assumptions about the identities of informers, informers' collusion is naturally indicated and considered in our protocol design. In addition, similar to the case ``{\em Honest $L_1$, Malicious $I_1$.} '', informers' collusion do not affect the possibility of guessing any secret ID $id_{xy}$. One potential effect is that multiple informers collude to submit only once in order to share higher total reward. However, such collusion is difficult among open population and thus have limited impact to the number of submissions. 
   \newpage
   \item {\em Malicious coalition $(L_1, L_2, I_1)$.}  This case is identical to the combination of  case ``{\em Malicious coalition $(L_1, L_2)$}'' and case ``{\em Malicious coalition $(L_1, I_1)$.}''.
   \item {\em Malicious coalition $(L_1, L_2, I_1,I_2)$.}  This case is impractical and analyzed in case {\em Malicious coalition $(L_1, I_1)$.}.
\end{itemize}
Thus, since we set $\lambda=128$, from all cases above, we can conclude that Argus contract can correctly indicate the infringer and tally the number of pirated copies, which is the interest of Owner. In other words, Proposition \ref{pro:1} is true. 
 \end{proof}

 \begin{proposition}
 If Licensee $L_1$ is honest, the interest of $L_1$ is guaranteed in $\mathsf{Game}^{Argus}_{O, L_1, L_2,I_1, I_2}$. 
 \end{proposition}
  \begin{proof}
For simplicity, we only present three extreme cases:  malicious $O$ or malicious coalition $(L_2,I_1, I_2)$ or malicious coalition $(O, L_2,I_1, I_2)$ (other cases can be analyzed similarly):
\begin{itemize}
\item {\em Malicious $O$ only}. As stated in Section \ref{sec:detect}, $O$ only has a possibility of $1/N$ to successfully accuse $L_1$. In addition, the incrimination is one-time since the process will be recorded by the ledger and discredit the owner. With a large $N$, the owner should tend to not accuse $L_1$.
\item {\em Malicious coalition $(L_2,I_1, I_2)$}. Similar to the case ``{\em Honest $L_1$, Malicious $I_1$.}'', the coalition can mislead the Owner to accuse $L_1$ with a possibility of $O(2^{-\lambda})$ which is negligible. 
\item {\em Malicious coalition $(O, L_2,I_1, I_2)$}. This equals to the combination of case  ``{\em malicious $(L_2,I_1, I_2)$}'' and case ``{\em Malicious $O$ only}'', which has a maximum possibility of $1/N$ to successfully incriminate the innocent $L_1$.
\end{itemize} 
Thus, the possibility to accuse honest $L_1$ is at most $1/N$ which is negligible with a large $N$ (e.g. $N=10,000$). Thus, we can conclude that the interest of $L_1$ is guaranteed. 
 \end{proof}
 
 \begin{proposition}
 If Informer $I_1$ is honest, the interest of $I_1$ is guaranteed in $\mathsf{Game}^{Argus}_{O, L_1, L_2,I_1, I_2}$. 
 \end{proposition}
  \begin{proof}
There are two ways to affect $I_1$'s interests: (1) to prevent $I_1$ from generating ProofOfLeakage (2) to decrease $I_1$'s reward of reporting a piracy. We will show that both ways are infeasible due to conflict of interests even if $(L_1, I_2)$ are malicious:
  \begin{itemize}
  \item {\em Malicious $L_1$}. This case is similar to the case ``{\em Malicious $L_1$, Honest $I_1$}'': since $L_1$ cannot remove the watermark and prevent $I_1$ from cooperating with $O$, the first way mentioned above is infeasible. In addition, since $L_1$ does not confess crimes via reporting, the reward value of  $I_1$ is not affected. 
   \item {\em Malicious $I_2$}. This case is similar to ``{\em Malicious $L_1$, Malicious $I_1$.}''. $I_2$ can neither interrupt $I_1$ to submit report nor reduce $I_1$'s reward via over-submission which conflicts with $I_2$'s interests (due to sybil-proofness).
    \item {\em Malicious coalition $(L_1, I_2)$}. This case equals to the combination of above two cases, which is infeasible due to either technical difficulties or conflict of interests. 
  \end{itemize}
 Therefore, the interest of $I_1$ is guaranteed\footnote{As addressed in Section \ref{sec:analysis}, we assume that the owner will not be malicious against informers since owner's interest is to obtain good-faith reports.}. 
 \end{proof}

\end{document}